\documentclass[11pt, letterpaper]{article} 
\usepackage[T1]{fontenc}
\usepackage{lmodern}
\usepackage[margin=1in]{geometry}

\usepackage{amsmath, amsthm,amssymb,amstext} 
\usepackage[usenames]{xcolor} 
\definecolor{darkgreen}{rgb}{0.0,0.3,0.0}
\definecolor{mygreen}{rgb}{0.0,0.6,0.0}
\definecolor{myorange}{cmyk}{0, 0.8808, 0.4429, 0.1412}
\usepackage[colorlinks=true,urlcolor=red,linkcolor=darkgreen,citecolor=darkgreen]{hyperref}
\usepackage{mathtools}
\usepackage{thmtools,thm-restate}

\usepackage[boxed]{algorithm2e}
\usepackage{todonotes}
\presetkeys{todonotes}{fancyline, color=blue!10}{}

\renewcommand{\upsilon}{\mu}

\title{Exponential Separation between Quantum Communication and Logarithm of Approximate Rank}
\date{}

\author{Makrand Sinha\thanks{CWI, the Netherlands. Supported by the Netherlands Organization for Scientific Research, Grant Number 617.001.351. {\tt makrand.sinha@cwi.nl}} 
\and 
Ronald de Wolf\thanks{QuSoft, CWI and University of Amsterdam, the Netherlands. Partially supported by ERC Consolidator Grant 615307-QPROGRESS and by QuantERA project QuantAlgo 680-91-034. {\tt rdewolf@cwi.nl}} 
}

\newcommand{\CM}{\mathcal{M}}

\newcommand{\CD}{\mathcal{D}}
\newcommand{\CH}{\mathcal{H}}

\newcommand{\CQ}{\mathcal{Q}}

\newcommand{\s}{\mathsf{s}}

\newenvironment{proofsk}{\nopagebreak
{\bf Proof Sketch:}}{ \qed \par \medskip}

\makeatletter
\def\thm@space@setup{\thm@preskip=5pt
\thm@postskip=0pt}
\makeatother

\newtheorem{theorem}{Theorem}[section]
\newtheorem{claim}[theorem]{Claim}
\newtheorem{proposition}[theorem]{Proposition}
\newtheorem{lemma}[theorem]{Lemma}

\newcommand{\norm}[1]{\left\|#1\right\|_1}

\newcommand{\papprox}[1]{\stackrel{#1}{\approx}}

\newcommand{\bits}{\{0,1\}}

\newcommand{\etal}{{\it et al.}}

\newcommand{\equationref}[1]{(\ref{#1})}

\newcommand{\Ex}[2]{\mathop{\mathbb{E}}\displaylimits_{#1}\left
[ #2 \right ]}
\newcommand{\BE}{\mathbb{E}}

\newcommand{\Y}{\overline{Y}}
\newcommand{\X}{\overline{X}}
\newcommand{\XR}{\X}
\newcommand{\YR}{\Y}

\newcommand{\bA}{A^*}
\newcommand{\bB}{B^*}

\newcommand{\CW}{\mathcal{W}}
\newcommand{\CE}{\mathcal{E}}
\newcommand{\CS}{\mathcal{S}}
\newcommand{\CT}{\mathcal{T}}
\newcommand{\eps}{\epsilon}
\newcommand{\omicron}{o}
\newcommand{\BP}{\mathbb{P}}
\newcommand{\sink}{\mathsf{sink}}

\providecommand{\ket}[2][]{|#2\rangle_{#1}}
\providecommand{\bra}[2][]{\langle #2 |_{#1}}

\providecommand{\ketbra}[3][]{| #2 \rangle \langle #3|_{#1}}
\newcommand{\Fi}[2]{F(#1,~#2)}

\newcommand{\Tr}{\mathsf{Tr}}

\newcommand{\pstate}[2]{#1^{(#2)}_{X\XR Y\YR ABC}}
\newcommand{\pstatealt}[3]{#1^{(#2)}_{#3}}
\newcommand{\pstatesup}[4]{#1^{(#2),#3}_{#4}}
\newcommand{\bstate}[2]{#1^{(#2)}}
\newcommand{\purify}[3]{\ket[#3]{#1^{(#2)}}}
\newcommand{\purifycond}[4]{\ket[#4]{#1^{(#2),#3}}}
\newcommand{\tens}[1]{%
  \mathbin{\mathop{\otimes}\limits_{#1}}%
}
\newcommand{\fixedvstate}[2]{{\iota}^{v,(#1)}_{#2}}

\newcommand{\allreg}{X\XR Y \YR ABC}
\newcommand{\ketalt}[2]{\ket[#2]{#1}}

\providecommand{\Div}[2]{\mathbf{D}\left(#1 \:||\:  #2\right)} 
\providecommand{\Inf}[3][]{\mathbf{I}_{#1}\left(#2:#3\right)} 
\providecommand{\Infc}[4][]{\mathbf{I}_{#1}\left(#2:#3|#4\right)}
\providecommand{\stat}[1]{\left\|#1\right\|_{\mathsf{tv}}}

\newcommand{\BX}{\mathbf{X}'_1}
\newcommand{\bX}{\mathbf{x_1}}
\newcommand{\BY}{\mathbf{Y'}_1}
\newcommand{\bY}{\mathbf{y_1}}
\newcommand{\alt}{\BX \BY}

\newcommand{\rest}{X_2 \overline{X_2} Y_2 \overline{Y_2} ABC}
\newcommand{\he}[2]{\mathfrak{h}\left( {#1},~{#2} \right)}

\newcommand{\I}{\mathbf{1}}

\numberwithin{figure}{section}

\begin{document}

\maketitle

\begin{abstract}
	 Chattopadhyay, Mande and Sherif~\cite{CMS18} recently exhibited a total Boolean function, the \emph{sink} function, that has polynomial approximate rank and polynomial randomized communication complexity. This gives an exponential separation between randomized communication complexity and logarithm of the approximate rank, refuting the log-approximate-rank conjecture. We show that even the \emph{quantum} communication complexity of the sink function is polynomial, thus also refuting the quantum log-approximate-rank conjecture. 
	 
	 Our lower bound is based on the fooling distribution method introduced by Rao and Sinha \cite{RS15} for the classical case and extended by Anshu, Touchette, Yao and Yu \cite{ATYY17} for the quantum case. We also give a new proof of the classical lower bound using the fooling distribution method.
\end{abstract}
	
\section{Introduction}

Communication complexity~\cite{KN96,RY18} is a basic model of distributed computing where one only cares about the resource of \emph{communication} between the various distributed parties doing the computation. This is a beautiful and fundamental computational model in its own right, and has many applications to other areas, in particular for lower bounds.
For concreteness consider the two-player communication complexity of some Boolean function $f:\bits^n\times\bits^n\to\bits$. Here Alice receives input $x\in\bits^n$ and Bob receives input $y\in\bits^n$, and they want to compute $f(x,y)$ with minimal communication between them.

Much research has gone into relating the (deterministic, randomized, nondeterministic, quantum, \ldots) communication complexity of~$f$ to its combinatorial or algebraic properties. In particular, we may consider the relation between the communication complexity and the \emph{rank} (over the reals) of the $2^n\times 2^n$ Boolean matrix $M_f$ whose entries are the values $f(x,y)$. Mehlhorn and Schmidt~\cite{mehlhorn&schmidt:lasvegas} showed that the log of this rank lower bounds the deterministic communication complexity of~$f$, and Lov\'asz and Saks~\cite{lovasz&saks:cc} conjectured that this lower bound is polynomially tight; in other words, that deterministic communication complexity is upper bounded by a polynomial in the logarithm of the rank of~$M_f$.
This \emph{log-rank conjecture} is one of the main open problems in communication complexity and remains wide open. On the one hand, the best upper bound on deterministic communication complexity in terms of rank is roughly the \emph{square root} of the rank~\cite{lovett:rootrank,lovett:advancesrank} (see also \cite{R14}). On the other hand, the biggest known gap between deterministic communication complexity and log-rank is only quadratic~\cite{GPW15}.

One may similarly consider the relation between \emph{randomized} communication complexity (say with private coin flips, and error probability $\leq 1/3$ on every input $x,y$) and log of the \emph{approximate} rank, which is the minimal rank among all matrices that approximate $M_f$ entrywise up to~$1/3$. The log of the approximate rank lower bounds randomized communication complexity (even \emph{quantum} communication complexity with unlimited prior entanglement~\cite{buhrman&wolf:qcclower}), and Lee and Shraibman~\cite[Conjecture~42]{lee&shraibman:survey} conjectured that this lower bound is polynomially tight. This is known as the \emph{log-approximate-rank} conjecture.
Until very recently, the biggest separation known between randomized communication complexity and the log of approximate rank was a fourth power~\cite{GJPW17}. 
But then, in an important breakthrough, Chattopadhyay, Mande and Sherif~\cite{CMS18} devised a function that refutes this conjecture.\footnote{Their function is a so-called \emph{XOR function}, of the form $f(x,y)=g(x\oplus y)$ for some $n$-bit Boolean function~$g$, and thus even refutes the special case of the log-approximate-rank conjecture restricted to XOR functions. The special case of the log-rank conjecture for such functions has received much attention recently~\cite{tsang:logrank,zhang:quantumxor,HHL:xor} (in part thanks to the fact that the rank of $M_f$ equals the \emph{Fourier sparsity} of~$g$), and remains open.}

Their function is as follows.
Let $n = \binom{t}{2}$. The function $\sink: \bits^n \to \bits$ is defined on the edges of the complete graph on the vertex set $[t]$.\footnote{We use $t$ for the number of vertices in the graph instead of $m$ as used in \cite{CMS18}.} For each edge $e \in \binom{[t]}{2}$, the corresponding input bit $z_e$ assigns an orientation to the edge $e$ (such an oriented complete graph is called a \emph{tournament}). The function $\sink(z)=1$ iff there is a vertex that is a \emph{sink} (i.e., that has no outgoing edges). Note that a tournament can have at most one sink, since the orientation of the edge between vertices $v$ and $w$ eliminates one of them as a possible sink. The communication problem is defined by Alice and Bob receiving the inputs $x,y \in\bits^n$ and wanting to compute the function $\sink(x \oplus y)$ where $x \oplus y$ is the bitwise parity. In other words, together they compute the sink function after putting the label $x_e \oplus y_e$ on the edge $e$.
With slight abuse of notation, we denote the $2n$-bit communication function by $\sink$ as well.

The approximate rank of the $2^n\times 2^n$ Boolean matrix $M_\sink$ associated to the sink problem is only polynomial in~$n$, which can be seen as follows.
Consider vertex $v\in[t]$, let $N(v)$ denote the set of edges incident on $v$ and let $x_{N(v)}$ (and $y_{N(v)}$) denote the projection of the input $x$ (and $y$) to the edges in $N(v)$. Let $z_{N(v)}\in\bits^{t-1}$ be the unique string of orientations that makes $v$ the sink of the graph. Note that $v$ is a sink in the tournament $x\oplus y$ iff $x_{N(v)}=y_{N(v)}\oplus z_{N(v)}$. The latter problem corresponds to a (shifted) equality problem on strings of $t-1$ bits, and it is well known that this problem has a cheap randomized private-coin protocol that uses $O(\log t)=O(\log n)$ bits of communication, that outputs~1 with probability~1 if $v$ is the sink in tournament $x\oplus y$, and outputs 1 with probability $\in[0,1/(3t)]$ if $v$ is not a sink. This in turns implies the existence of a $2^n\times 2^n$ matrix $M_v$ of rank polynomial in $n$, whose $(x,y)$-entry is~1 if $v$ is the sink in $x\oplus y$, and whose $(x,y)$-entry is~$\in[0,1/(3t)]$ if $v$ is not a sink. Thanks to the fact that at most one of the $t$ vertices is a sink, we can now get a good entry-wise approximation of $M_\sink$ by just adding up all the $M_v$-matrices over all $v \in [t]$: the resulting matrix $\widetilde{M}=\sum_{v=1}^t M_v$ will have $(x,y)$-entry $\in[0,1/3]$ whenever $x\oplus y$ has no sink, and will have $(x,y)$-entry in $[1,4/3]$ whenever $x\oplus y$ has a sink (if $v$ is the sink in $x\oplus y$, then $M_v$ contributes~1 to the entry $\widetilde{M}_{xy}$, and the other $M_w$'s together contribute at most $1/3$). By subadditivity of rank, the rank of $\widetilde{M}$ is at most the sum of the ranks of the $M_v$'s, which is polynomial in~$n$. Hence the log of the approximate rank of $M_{\sink}$ is $O(\log n)$.
In contrast, Chattopadhyay \etal\ show that the randomized communication complexity of the sink function is exponentially bigger:

\begin{theorem}[\cite{CMS18}]\label{th:CMS}
The $1/3$-error randomized communication complexity of  the function $\sink$ on  $n = \binom{t}{2}$ bits is $\Omega(t) = \Omega(\sqrt{n})$.
\end{theorem} 

This lower bound is optimal even for \emph{deterministic} protocols: by looking at one edge, Alice and Bob can rule out one vertex from being a sink. Proceeding this way, they read $t-1$ edges until they have eliminated all but one vertex $v$ from being a sink. At this point, they look at the $t-1$ edges incident to $v$, and find out if $v$ is a sink or not. This gives an $O(t)$-bit deterministic communication protocol, since the parties exchange two bits per edge.

This separation refutes the log-approximate-rank conjecture, showing that randomized communication complexity is not always upper bounded by polylog of the approximate rank. However, \emph{quantum} communication complexity can be much smaller than randomized communication complexity: polynomial gaps are known for some total functions~\cite{BuhrmanCleveWigderson98,aaronson&ambainis:searchj,anshuetal:sepcheatsheets} and exponential gaps are known for some partial functions~\cite{raz:qcc,klartag&regev:lowerbound}. Thus one might still entertain the weaker conjecture that quantum communication complexity is upper bounded by polylog of the approximate rank, and indeed Lee and Shraibman~\cite[Conjecture~57]{lee&shraibman:survey} made this conjecture explicitly. 
Prior to this work, the biggest separation known between quantum communication complexity and log of the approximate rank, was only quadratic~\cite{anshuetal:sepapproxrank}. Indeed, one of the main problems left open by Chattopadhyay \etal\ asks about the quantum communication complexity of the $\sink$ function. If this is large then it would refute the quantum log-approximate-rank conjecture, but if it is small then it would provide the first superpolynomial separation between quantum and classical communication complexity for a total function.
We answer their open question by proving a polynomial lower bound on the quantum communication complexity of the $\sink$ function, thus refuting the quantum log-approximate-rank conjecture:

\begin{theorem}\label{th:qlogrank}
The $1/3$-error quantum communication complexity of  the function $\sink$ on  $n = \binom{t}{2}$ bits is $\Omega(t^{1/3}) = \Omega(n^{1/6})$.
\end{theorem} 

As Chattopadhyay \etal\ noted, the quantum communication complexity of the sink function is polynomially smaller than the randomized complexity: using Grover's algorithm~\cite{grover:search} to search for a sink, combined with an efficient low-error equality protocol to test whether a specific vertex is a sink, one gets an $\widetilde{O}(\sqrt{t})$-qubit protocol. We suspect that this upper bound is tight up to the log-factor, and that our quantum lower bound should be improvable.

\paragraph{Independent Work.} 
\begin{sloppypar}
In independent and simultaneous work, Anshu, Boddu and Touchette \cite{AGT18} obtained the same $\Omega(t^{1/3})$ lower bound using a reduction to quantum information complexity of the equality function, but our techniques to prove Theorem \ref{th:qlogrank} are different, as we describe below.
\end{sloppypar}

\subsection*{Proof Outline}

Our approach to proving Theorem~\ref{th:qlogrank} is to first give an alternate and arguably simpler proof of Theorem~\ref{th:CMS} using the fooling distribution method (and other tools) introduced by Rao and the first author in \cite{RS15}, and then we show that the same approach can be used to give a (weaker) quantum lower bound using tools from a paper by Anshu, Touchette, Yao and Yu \cite{ATYY17}, which generalized some of the techniques used in \cite{RS15} to the quantum setting. Our proofs are relatively straightforward and short given the tools in these papers. Below we give a high-level outline.

Let us look at the classical case first. To prove a lower bound on the randomized communication complexity, it suffices to give a distribution on the inputs that is hard for deterministic protocols. Let $p_0(X,Y)$ denote the uniform distribution on $0$-inputs to $\sink$ and $p_1(X,Y)$ denote the uniform distribution on $1$-inputs to $\sink$. Our hard distribution for deterministic protocols will be the distribution which samples from $p_0(X,Y)$ with probability $\frac12$ and from $p_1(X,Y)$ with probability $\frac{1}{2}$. Note that the messages of any low-error protocol look very different under these two distributions: $p_0(M)$ and $p_1(M)$ have statistical distance close to $1$, where $p_b(M)$ denotes the distribution induced on the messages under $p_b(X,Y)$ for $b\in \bits$. 

To show that this is a hard distribution for deterministic protocols, we show that there is another distribution $u(X,Y)$ such that for any protocol with communication at most $\eps t$, the induced message distribution $u(M) \approx p_0(M)$ as well as $u(M) \approx p_1(M)$, where $\approx$ denotes closeness in statistical distance. This in turn implies that $p_0(M) \approx p_1(M)$ for small-communication protocols, giving us a lower bound on communication. Such a distribution $u(X,Y)$ is called a \emph{fooling} distribution. 

The fooling distribution $u(X,Y)$ for $\sink$ will just be the uniform distribution on $\bits^{n+n}$. Note that under the uniform distribution $u(X,Y)$, the function $\sink$ takes value $0$ with probability $1-2^{-\Omega(t)}$, and since $p_0(X,Y) = u(X,Y|\sink=0)$, the input distributions $p_0(X,Y)$ and $u(X,Y)$ are already very close in statistical distance, and so are the corresponding distributions on the messages. The interesting part is to argue that the message distribution $p_1(M) \approx u(M)$ even though the respective input distributions $p_1(X,Y)$ and $u(X,Y)$ are actually very far apart. For this purpose, let us first note that the distribution $p_1(X,Y)$ can be generated from $u(X,Y)$ by first picking a uniformly random vertex $v$ as the sink and conditioning on the event that $X_{N(v)} = Y_{N(v)} \oplus z_{N(v)}$ (recall that $N(v)$ is the set of edges incident on $v$, $X_{N(v)}$ and $Y_{N(v)}$ are projections of $X$ and $Y$ to the edges in $N(v)$, and $z_{N(v)}$ is the unique string that encodes the orientation of the edges for which vertex $v$ is the sink).

To argue that $p_1(M) \approx u(M)$, first one can use Shearer's inequality (see Lemma \ref{lemma:shearer}) to conclude that under the distribution $u(X,Y)$, the messages $M$ reveal only a small amount of information about $X_{N(v)}$ and $Y_{N(v)}$ for a random vertex $v$. In particular, since an edge appears in $N(v)$ with probability $2/t$ for a random $v$, one would expect $M$ to reveal at most $(2/t) \cdot |M| \le \eps$ bits of information about $X_{N(v)}$ and $Y_{N(v)}$ each (this is also the reason for working with the fooling distribution: since all the inputs are independent of each other, one may use Shearer's inequality). Now to relate the fooling distribution $u(X,Y)$ to the input distribution $p_1(X,Y)$ we need to condition on the event $X_{N(v)} = Y_{N(v)} \oplus z_{N(v)}$. A lemma from \cite{RS15} (see Lemma \ref{lemma:condition} in Section \ref{sec:classical}) exactly captures this situation and says that conditioning on such a collision event, when the messages reveal little information about the colliding variables, does not change the distribution of the messages too much, so we can conclude that $p_1(M) \approx u(M)$.

The proof for the quantum case proceeds more or less analogously. It is still true that the output of a low-error quantum protocol must look very different under distributions supported only on 0-inputs and 1-inputs respectively. We show that $u(X,Y)$ is still a fooling distribution for small-communication quantum protocols. As in the classical case, it is easy to argue using a quantum version of Shearer's inequality (see Lemma \ref{lemma:qshearer}) that small-communication quantum protocols do not reveal too much information about $X_{N(v)}$ and $Y_{N(v)}$ for a random vertex $v$ under the fooling distribution $u(X,Y)$. To condition on the collision event $X_{N(v)} = Y_{N(v)} \oplus z_{N(v)}$, we use a lemma from \cite{ATYY17} (see Lemma \ref{lemma:qcond} in Section \ref{sec:quantum}) which allows us to argue that for a typical vertex $v$, conditioning on the collision event does not change the output too much. So, it must be the case that for a small-communication quantum protocol, the output on an input distribution where $v$ is the sink (for a typical $v$) must be close to the output when the input distribution is $p_0(X,Y)$. This implies that small-communication quantum protocols for the $\sink$ function must have large error. 

\paragraph{Organization.} We introduce preliminaries on information theory, quantum information theory and communication complexity in the next section (Section \ref{sec:prelims}). Section \ref{sec:classical} contains the proof described above for the classical case. The quantum lower bound is given in Section \ref{sec:quantum}.

\section{Preliminaries} \label{sec:prelims}

\subsection{Classical Probability Theory}

\subsubsection*{Probability Spaces and Variables}

Throughout this paper, $\log$ denotes the logarithm taken in base two. We use $[k]$ to denote the set $\{1,2,\dotsc, k\}$ and $[k]^{<n}$ to denote the set of all strings of length less than $n$ over the alphabet $[k]$, including the empty string. The notation $|z|$ denotes the length of the string $z$.

Random variables are denoted by capital letters (e.g.\ $A$) and values they attain are denoted by lower-case letters (e.g.\ $a$). Events in a probability space will be denoted by calligraphic letters (e.g.\ $\CE$). Given $a = (a_1,a_2, \dotsc, a_n)$, we write $a_{\leq i}$ to denote $a_1,\dotsc, a_i$. We define $a_{< i}$ similarly. We write $a_S$ to denote the projection of $a$ to the coordinates specified in the set $S \subseteq [n]$. 

Given a probability space $p$ and a random variable $A$ in the underlying sample space, we use the notation $p(A)$ to denote the probability distribution of the variable $A$ in the probability space~$p$. We will often consider multiple probability spaces with the same underlying sample space, so for example $p(A)$ and $q(A)$ will denote the distribution of the random variable $A$ under the probability spaces $p$ and $q$, respectively, with the underlying sample space of $p$ and $q$ being the same. We write $p(A|b)$ to denote the distribution of $A$ conditioned on the event $B=b$. We write $p(a)$ to denote the number $\BP_p[A=a]$ and $p(a|b)$ to denote the number $\BP_p[A=a|B=b]$. Given a distribution $p(A,B,C,D)$, we write $p(A,B,C)$ to denote the marginal distribution on the variables $A,B,C$. We often write $p(AB)$ instead of $p(A,B)$ for conciseness of notation. Similarly, $p(a,b,c)$ will denote the probability according to the marginal distribution $p(A,B,C)$ and we will often write it as $p(abc)$ for conciseness. 

If $\CW$ is an event, we write $p(\CW)$ to denote its probability according to $p$. For two events $\CW$ and $\CW'$, the probability of their intersection $\CW \cap \CW'$ is denoted by $p(\CW,\CW')$ . Given a probability space $p$ and a random variable $A$, when we write $A \in \CW$ for an event $\CW$ we only consider events in the space of values taken by the variable $A$.

Given a fixed value $c$, we denote by $\Ex{p(b|c)}{g(a,b,c)} := \sum_{b} p(b|c) \cdot g(a,b,c)$, the expected value of the function $g(a,b,c)$ under the distribution $p(B|c)$. If the probability space $p$ is clear from the context, then we will just write $\Ex{b|c}{g(a,b,c)}$ to denote the expectation. For a Boolean function $h(a,b)$ and a probability distribution $p(A,B)$, we use $\I[h(a,b)=0]$ to denote the indicator function for the event $h(a,b)=0$, and we write $p(h=0) := \BE_{p(ab)}[\I[h(a,b)=0]]$ as the probability that $h$ is $0$ under inputs drawn from $p$.

We write $A-M-B$ as a shorthand to say that the random variables $A$, $M$ and $B$ form a \emph{Markov chain}, or in other words, that $A$ and $B$ are independent given $M$: $p(amb) = p(m) \cdot p(a|m) \cdot p(b|m)$ for every $a,b,m$. 

{To illustrate the notation, consider the following example. Let $A \in \bits^2$ be a uniformly distributed random variable in a probability space $p$. Then, $p(A)$ is the uniform distribution on $\bits^2$, and if $a=(0,0)$ then $p(a) = 1/4$. Let $A_1$ and $A_2$ denote the first and second bits of $A$, then if $B = A_1 + A_2 \bmod 2$, then when $b=1$, $p(A|b)$ is the uniform distribution on $\{(0,1),(1,0)\}$. If $a = (1,0)$ and $b=1$, then $p(a|b)=1/2$ and $p(a,b) = 1/4$. If $\CE$ is the event that $A_1=B$, then $p(\CE)=1/2$. Let $q(A)=p(A|\CE)$, then $q(A)$ is the uniform distribution on $\{(0,0),(1,0)\}$ and $q(A_2)$ is the distribution over the sample space $\bits$ which takes the value $0$ with probability $1$.}

\subsubsection*{Statistical Distance}

For two distributions $p(A),q(A)$, the \emph{statistical} (or \emph{total variation}) \emph{distance} $\stat{p(A) - q(A)}$ between them is defined to be $\stat{p(A) - q(A)} =  \max_\CQ \left(p(A \in \CQ) - q(A \in \CQ)\right)$ where $\CQ$ ranges over all events. The following propositions are easy to prove.

\begin{proposition} 
	\label{prop:statdef}
	$\stat{p(A) - q(A)} = \frac12 \sum_{a} |p(a) - q(a)| = \sum_{a: p(a)>q(a)} (p(a) - q(a)).$
\end{proposition}

We say $p(A)$ and $q(A)$ are \emph{$\eps$-close} if $|p(A)-q(A)| \leq \epsilon$ and we write it as $p(A) \papprox{\epsilon} q(A)$.

\begin{proposition} \label{prop:convexstat}
	If $p(AB), q(AB)$ are such that $p(A) = q(A)$, then 
	\[\stat{p(B) - q(B)} = \Ex{p(a)}{\stat{p(B|a) - q(B|a)}}.\]
\end{proposition}

\begin{lemma} \label{prop:statcond}
	If $\CE$ is an event such that $p(\CE) = 1 -\delta$, then $\stat{p(A|\CE) - p(A)} = \delta.$
\end{lemma}
\begin{proof}
	Note that for any $a \notin \CE$, $p(a|\CE) = 0$ and for $a \in \CE$, using Bayes' rule, we get that 
	\begin{align} \label{eqn:bayes}
	\ p(a|\CE) = \frac{p(a, \CE)}{p(\CE)} = \frac{p(a)}{p(\CE)} = \frac{p(a)}{1- \delta}.
	\end{align}
	By Proposition \ref{prop:statdef}, we have that 
	\begin{align*}
		\ \stat{p(A|\CE) - p(A)} & =\frac12 \sum_{a \in \CE} |p(a|\CE) - p(a)| + \frac12 \sum_{a \notin \CE} |p(a|\CE) - p(a)| = \frac12 \sum_{a \in \CE} |p(a|\CE) - p(a)|+ \frac{\delta}{2} \\
		\ & \stackrel{\eqref{eqn:bayes}}{=} \frac12 \sum_{a \in \CE} \left(\frac{p(a)}{1-\delta} - p(a)\right) + \frac\delta{2} = \frac12  \cdot \frac\delta{1-\delta} \cdot p(\CE) + \frac\delta2 = \delta,
	\end{align*}
	where the second inequality follows from \eqref{eqn:bayes}.
\end{proof}

\subsubsection*{Divergence and Mutual Information}
The \emph{divergence} between distributions $p(A)$ and $q(A)$ is defined to be 
\[
\Div{p(A)}{q(A)} = \sum_a p(a) \log \frac{p(a)}{q(a)}.
\] 
In a probability space $p$, the \emph{mutual information} between $A,B$ conditioned on $C$ is defined as 
\begin{align*}
\ \Infc[p]{A}{B}{C} &= \Ex{p(bc)}{ \Div{p(A|bc)}{p(A|c)}} \\
\                   &= \Ex{p(ac)}{ \Div{p(B|ac)}{p(B|c)}} = \sum_{a,b,c} p(abc) \log \frac{p(a|bc)}{p(a|c)}.
\end{align*}

\subsubsection*{Basic Information Theory Facts}

The proofs of the following basic facts can be found in the book by Cover and Thomas \cite{CT06}. {In the following, $p$ and $q$ are probability spaces (over the same sample space), and $A,B$ and $C$ are random variables on the underlying sample space.}

\begin{proposition} \label{proposition:divpositive} $\Div{p(A)}{q(A)} \geq 0$.
\end{proposition}

\begin{proposition} \label{proposition:infoupper} If $A \in \{0,1\}^\ell$, then $\Inf[p]{A}{B} \leq \ell$.
\end{proposition}

\begin{proposition}[Pinsker's Inequality] \label{proposition:pinsker} 
$$\stat{p(A)- q(A)}^2 \leq \dfrac{\ln 2}{2} \cdot \Div{p(A)}{q(A)} \le \Div{p(A)}{q(A)}.$$
\end{proposition}

\begin{lemma}[Shearer's Inequality \cite{GKR14}]\label{lemma:shearer}
 Let $A= (A_1,\dotsc, A_n)$ where the $A_i$'s are mutually independent. Let $M$ be another random variable and $S \subseteq [n]$ be a random set independent of $A$ and $M$, such that $p(i \in S) \le \mu$ for every $i \in [n]$. Then, we have 
$$\Infc[p]{A_S}{M}{S} \le \mu \cdot \Inf[p]{A}{M}. $$
\end{lemma}

\subsection{Classical Communication Complexity} \label{subsec:cc}

The {\em communication complexity} of a protocol is the maximum number of bits that may be exchanged by the protocol.  Communication protocols may use \emph{shared randomness} and henceforth we will refer to such protocols as randomized protocols. We say a randomized protocol computing a Boolean function $f(x,y)$ has error $\delta$, if for every input, the protocol outputs the correct answer with probability at least $1 - \delta$, where the probability is over the shared randomness. 

We briefly describe some basic properties of communication protocols that we need. For more details see the textbooks \cite{KN96} or \cite {RY18}. For a deterministic protocol $\pi$, let $\pi(x,y)$ denote the messages of the protocol on inputs $x,y$. For any transcript $m$ of the protocol, define the events:
\begin{align*}
	\CS_m &= \{x \mid \exists y \text{ such that } \pi(x,y) = m\}, & \CT_m &= \{y \mid \exists x \text{ such that } \pi(x,y) = m\}.
\end{align*}
We then have:
\begin{proposition}[Messages Correspond to Rectangles]
	\label{prop:rec}
	If $m$ is a transcript and $x,y$ are inputs to a deterministic protocol $\pi$, then, $\pi(x,y) = m  \iff x \in \CS_m \wedge y \in \CT_m$.
\end{proposition}
\vspace*{2pt}
Proposition \ref{prop:rec} implies:
\begin{proposition}[Markov Property of Protocols]\label{prop:markov}
	Let $X$ and $Y$ be random inputs to a deterministic protocol and let $M$ denote the messages of this protocol. If $X$ and $Y$ are independent then $X - M -Y$.
\end{proposition}
\vspace*{2pt}

\begin{lemma}[Errors and Statistical Distance] 
	\label{lemma:error}
	Let $h(x,y)$ be a boolean function and $p(X,Y)$ be a distribution such that $p(h=0) = p(h=1) = \frac12$. If $\pi$ is a deterministic protocol with messages $M$ that computes $h$ with error $\delta$ on the distribution $p(XY)$, then $|p(M|h=0)-p(M|h=1)|\ge 1 - 2\delta$. 
\end{lemma}
\begin{proof}
	Since $|p(M|h=0)-p(M|h=1)|=\max_{\CQ} (p(M\in \CQ|h=0) - p(M \in \CQ|h=1))$ it suffices to exhibit an event $\CQ$ such that $p(M\in \CQ|h=0) - p(M \in \CQ|h=1) = 1-2\delta$.  Let $\CM_0$ denote the event that the protocol outputs a zero. Then, since $p(h=0)=p(h=1)=\frac12$, writing the probability of success in terms of $\CM_0$, we have
	\[ 1 - \delta = \frac{p(M\in \CM_0|h=0)}{2} + \frac{1- p(M \in \CM_0|h=1)}{2} = \frac12 + \frac{p(M \in \CM_0|h=0) - p(M \in \CM_0|h=1)}{2}.\]
On rearranging, the above gives us that $p(M \in \CM_0|h=0) - p(M \in \CM_0|h=1) = 1-2\delta$ and hence the statistical distance must be at least $1-2\delta$.
\end{proof}

\subsection{Quantum Information Theory}

Here we briefly state the facts we need from quantum information theory. For details, see the textbooks \cite{W13} or \cite{W18}.

\subsubsection*{Quantum States and Measurements}

Overloading the notation, we use capital letters $A, B$, etc.\ to represent registers and use $\CH_A, \CH_B$, etc.\ to denote the associated Hilbert spaces. As before, given registers $A = A_1,\ldots, A_n$ and a set $S \subseteq [n]$, we will use $A_S$ to denote the sequence of registers $\{A_i\}_{i \in S}$. For any register $A$, $|A| = \lceil\log (\dim \CH_A)\rceil$ denotes the number of qubits in $A$. Given a Hilbert space $\CH_A$, we use $\{\ket[A]{a}\}$ to denote a canonical orthonormal basis, and if $A$ is a single-qubit register we use $\{\ket[A]{0}, \ket[A]{1}\}$ to denote the computational basis for the Hilbert space $\CH_A$. We write $U_A$ to denote a unitary acting on the Hilbert space $\CH_A$ corresponding to a register $A$. 

A \emph{density operator} on $\CH_A$ is a linear operator from $\CH_A$ to $\CH_A$ that is positive semi-definite and has a unit trace. The set of all density operators on a Hilbert space $\CH_A$ will be denoted by $\CD(\CH_A)$. Since a linear operator on a finite-dimensional Hilbert space can be described equivalently with a matrix representation, we will use these notions interchangeably. 

A \emph{(quantum) state} $\rho_A$ on a register $A$ is a density operator on $\CH_A$. A state $\rho_A$ is called \emph{pure} if it has rank 1. For a unit vector $\ket[A] \psi \in \CH_A$ (viewed as a column vector), we denote by $\bra[A]\psi$ its adjoint (a row vector), and by $\psi_A$ the corresponding state $\ketbra[A] \psi \psi $, but we will also sometimes use the vector $\ket[A] \psi$ to refer to the corresponding pure state. A classical distribution $p(A)$ can be viewed as the diagonal state $\sum_a p(a)\ketbra[A]{a}{a}$ and vice versa, so we will refer to any diagonal state as a classical state.

We use $\rho_A \otimes \sigma_B$ to denote the tensor product of $\rho_A$ and $\sigma_B$ on the Hilbert space $\CH_A \otimes \CH_B$.  We adopt the convention of omitting Identity operators from a tensor product: instead of $U_R \otimes I_A$ or $\bra[R]{r} \otimes I_A$, we write $U_R$ or $\bra[R]{r}$ since the subscripts will convey the necessary information.

A state $\rho_{XA}$ is called a \emph{classical-quantum} state with $X$ being the classical register if it is of the form $\rho_{XA} = \sum_{x} p(x) \ketbra[X]{x}{x} \otimes \rho^x_A$ where $p(X)$ is a classical probability distribution and $\rho^x_A$ is a state on the register $A$. 

Given a linear operator $M_{AB}$ on $\CH_A \otimes \CH_B$, the \emph{partial trace} of $M_{AB}$ over $A$ is defined as 
	\[ \Tr_A(M_{AB}) = \sum_{a} \bra[A]{a}M_{AB}\ket[A]{a}.\]
	
The partial trace operation is linear: $\Tr_A(M_{AB} + M'_{AB}) = \Tr_A(M_{AB}) + \Tr_A(M'_{AB})$ and satisfies the following identities: $\Tr_A (M_A \otimes M_B) = \Tr_A(M_A) M_B$ and $\Tr_A (U_B M_{AB}) = U_B \Tr_A(M_{AB})$.

With the above, we can define the notion of a marginal or reduced state: for a bipartite state $\rho_{AB}$, the \emph{marginal state} $\rho_B$ on the register $B$ is defined as $\rho_B := \Tr_A(\rho_{AB})$. Note that if we have a classical quantum state $\rho_{XA}$, then the marginal state $\rho_X$ is a classical state. 

Given a state $\rho_A$ we can always consider it as a marginal of a pure state $\rho_{EA} = \ketbra[EA]{\rho_{EA}}{\rho_{EA}}$ on a larger system. Such a state $\ket[EA]{\rho_{EA}}$ is called a \emph{purification} of $\rho_A$. We will adopt the convention of using the same Greek letters to denote the purification: if we say that $\ket[EA]{\rho_{EA}}$ is a purification with reference register $E$, then it is a purification of the state $\rho_A$, that is, $\rho_{A} = \Tr_E(\ketbra[EA]{\rho_{EA}}{\rho_{EA}})$. Given a classical state $\rho_X = \sum_x p(x) \ketbra[X]{x}{x}$, we define $\sum_{x} \sqrt{p(x)}\ket[X]{x}\ket[X]{x}$ to be its \emph{canonical} purification.

A \emph{positive operator valued measurement} (POVM) is a collection $\{\Lambda_i\}_i$ of linear operators acting on a Hilbert space $\CH_A$ such that for each $i$, the operator $\Lambda_i$ is positive semi-definite, and $\sum_i \Lambda_i = I_{A}$. The probability that the outcome of applying a POVM on a quantum state $\rho_A \in \CD(\CH)$ is $j$ is given by $\Tr(\Lambda_j \rho_A)$. Given a single-qubit register $A$, we will specifically be interested in measurement in the computational basis, which corresponds to the POVM $\{\ketbra[A]{0}{0}, \ketbra[A]{1}{1}\}$. Given a state $\rho_A \in \CD(\CH_A)$, the probability that the measurement outcome is the bit $b \in \bits$ is $\Tr(\ketbra[A]{b}{b}\rho_A)$.

We say that $U_{XA}$ is a unitary with $X$ as a \emph{control} register if $U_{XA} = \sum_x \ketbra[X]{x}{x} \otimes U^x_A$ for some $U^x_A$'s. Also, note that in this case $U^\dag_{XA}$ is a unitary controlled by $X$ as well.

\subsubsection*{Distance Measures}

Recall that the trace norm $\norm{M}$ of a matrix $M$ is defined as $\norm{M} = \Tr \sqrt{M^\dag M}$. Equivalently, $\norm{M}$ is the sum of the singular values of~$M$. Then, the \emph{trace distance} between two quantum states $\rho_A$ and $\sigma_A$ is defined as $\norm{\rho_A-\sigma_A}$.  We say two states $\rho_A$ and $\sigma_A$ are \emph{$\eps$-close} in trace norm if $\norm{\rho_A-\sigma_A} \le \eps$, and write this as $\rho_A \papprox{\eps} \sigma_A$.

The \emph{fidelity} between two quantum states is defined as $\Fi{\rho_A}{\sigma_A} = \norm{\sqrt{\rho_A}\sqrt{\sigma_A}}$ (note that some papers define fidelity as the square of our definition). If $\rho_A$ and $\sigma_A$ are pure states, then their fidelity is just the absolute value of the inner product of the corresponding vectors. The \emph{Hellinger distance} between the states is $\he{\rho_A}{\sigma_A} = \sqrt{1 - \Fi{\rho_A}{\sigma_A}} = \sqrt{1 - \norm{\sqrt{\rho_A}\sqrt{\sigma_A}}}$. If $\psi_A = \ketbra[A]{\psi}{\psi}$ is a pure state, for brevity we will sometimes write $\he{\ket[A]{\psi}}{\ket[A]{\sigma}}$ (or $\Fi{\ket[A]{\psi}}{\ket[A]{\sigma}}$) to mean $\he{\psi_A}{\sigma_A}$ (or $\Fi{\psi_A}{\sigma_A}$). The Hellinger distance is a metric and in particular satisfies the triangle inequality: $\he{\rho_A}{\sigma_A} \le \he{\rho_A }{\psi_A} + \he {\psi_A}{\sigma_A}$.

The trace distance and Hellinger distance are both invariant under applying unitaries and decrease under taking marginals: 
\begin{proposition} \label{prop:invar} Given unitaries $U_A$ and $V_A$, it holds that 
	\[ \norm{U_A(\rho_A - \sigma_A)V_A^\dag} = \norm{\rho_A - \sigma_A} \text{ and } \he{U_A \rho_A V_A^\dag}{U_A \sigma_A V_A^\dag} = \he{\rho_A}{\sigma_A}. \]
\end{proposition}

\begin{proposition} \label{prop:marginal-trace}
	$\norm{\rho_A - \sigma_A} \le \norm{\rho_{AB} - \sigma_{AB}} \text{ and } \he{\rho_A}{\sigma_A} \le \he{\rho_{AB}}{\sigma_{AB}}.$\\
\end{proposition}

The Hellinger and trace distance are related in the following way:

\begin{proposition} \label{prop:he-and-trace}
	For quantum states $\rho_A$ and $\sigma_A$, it holds that
	\[ \he{\rho_A}{\sigma_A}^2 \le \frac12 \norm {\rho_A - \sigma_A} \le \sqrt{2}\, \he {\rho_A} {\sigma_A}. \]
\end{proposition}

The trace distance normalized by $2$ is the largest probability difference a POVM could produce between the two states, which is the quantum generalization of total variation distance:
\begin{proposition} \label{prop:tracepovm}
	For states $\rho_A$ and $\sigma_A$ in $\CD(\CH_A)$, it holds that 
	\[ \frac12 \norm{\rho_A - \sigma_A} = \max_\Lambda \Tr(\Lambda(\rho_A - \sigma_A)),\]
	where $\Lambda$ ranges over all positive semi-definite operators over $\CH_A$ that have eigenvalues at most one.
\end{proposition}

\begin{proposition}[Uhlmann's Theorem] \label{prop:uhlmann}
	Let $\ket[EA]{\rho}$ and $\ket[EA]{\sigma}$ be pure states. Then, we have
	\begin{align*}
	\ \Fi{\rho_A}{\sigma_A} &= \max_{U_E} \Fi{U_E\ket[EA]{\rho}}{\ket[EA]{\sigma}}, \text{ or equivalently, } \\
	\	\ \he{\rho_A}{\sigma_A} &= \min_{U_E} \he{U_E\ket[EA]{\rho}}{\ket[EA]{\sigma}},
	\end{align*}
	where $U_E$ ranges over all unitaries acting on the register $E$.\\
\end{proposition} 

The unitary $U_E$ which minimizes the Hellinger distance in Uhlmann's theorem is the one for which $\sqrt{\rho_E}\sqrt{\sigma_E} U_E$ is positive semidefinite (such a unitary is always guaranteed to exist) but we will only need the following simple case:

\begin{proposition}\label{prop:control}
Let $p(X,Y)$ and $q(X,Y)$ be distributions such that $p(X)=q(X)$. Then for the quantum states $\ketalt{\rho}{X\X Y\Y} = \sum_{xy} \sqrt{p(x,y)}\ketalt{xxyy}{X\X Y\Y}$ and $\ketalt{\sigma}{X\X Y\Y} = \sum_{xy} \sqrt{q(x,y)}\ketalt{xxyy}{X\X Y\Y}$, there exists a unitary $W_{XY \Y}$ with $X$ as a control register such that $W_{XY \Y} \ketalt{\rho}{X\X Y\Y} =  \ketalt{\sigma}{X\X Y\Y}$.\\
\end{proposition}

The above is a special case of Uhlmann's Theorem as $\rho_{\X} = \sigma_{\X}$ but one can explicitly take $W_{XY \Y} = \sum_{x} \ketbra[X]{x}{x} \otimes U^x_{Y \Y}$ where $U^x_{Y\Y}$ is any unitary that maps the vector  $\sum_{y} \sqrt{p(x,y)} \ket[Y\Y]{yy}$ to $\sum_{y} \sqrt{q(x,y)} \ket[Y\Y]{yy}$.

\subsubsection*{Quantum Divergence and Mutual Information}

The \emph{divergence} (or \emph{relative entropy}) between two quantum states $\rho_A, \sigma_A \in \CD(\CH_A)$ is defined as 
\[ \Div{\rho_A}{\sigma_A} = \Tr(\rho_A \log \rho_A) - \Tr(\rho_A \log \sigma_A).\]
Note that the divergence between two states $\rho_A$ and $\sigma_A$ is always non-negative, and equal to zero iff $\rho_A=\sigma_A$.
The \emph{quantum mutual information} of the bipartite state $\rho_{AB}$ is defined as 
\begin{equation} \label{prop:qdiv-mutinf}
	\ \Inf[\rho]{A}{B} = \Div{\rho_{AB}}{\rho_{A} \otimes \rho_B}. 
\end{equation}
For a tripartite quantum state $\rho_{ABC} \in \CD(\CH_A \otimes \CH_B \otimes \CH_C)$, the \emph{conditional quantum mutual information} is defined as $\Infc[\rho]{A}{B}{C} = \Inf[\rho]{A}{BC} - \Inf[\rho]{A}{C}$.
For empty $C$, this equals the definition of mutual information in \eqref{prop:qdiv-mutinf}. 

It follows from the non-negativity of divergence that quantum mutual information is also non-negative, but it turns out that conditional mutual information is non-negative as well:
\begin{proposition}[Strong subadditivity]\label{prop:strong-subadd} 
$\Infc[\rho]{A}{B}{C}  \ge 0$.
\end{proposition}

\begin{proposition}[Chain Rule]
$  \Inf[\rho]{A}{BC} = \Inf[\rho]{A}{C} + \Infc[\rho]{A}{B}{C}.$
\end{proposition}

\begin{proposition} \label{prop:qmutinfbound}
$  \Infc[\rho]{A}{B}{C} \le 2\min\{|A|,|B|\}.$
\end{proposition}
\begin{proposition} \label{prop:qind} If $\rho_{AB} = \rho_A \otimes \rho_B$, then $\Inf[\rho]{A}{B} = 0$.
\end{proposition}

\subsubsection*{Basic Lemmas about Divergence and Mutual Information}

Below $\rho_{ABC}, \sigma_{ABC} \in \CD(\CH_A \otimes \CH_B \otimes \CH_C)$ and $U_B$ is a unitary acting on $B$.

\begin{proposition}[Pinsker's inequality]\label{prop:quantumpinsker}
	$ \frac18 \norm{\rho_A - \sigma_A}^2 \le \he{\rho_A}{\sigma_A}^2 \le \Div{\rho_A}{\sigma_A}.$\\
\end{proposition}
The proposition below says that mutual information does not change under local operations:
\begin{proposition}\label{prop:local}
If $\sigma_{ABC} = U_B~ \rho_{ABC}~ U_B^\dag$, then  $\Inf[\sigma]{A}{BC} = \Inf[\rho]{A}{BC}.$\\
\end{proposition}

Furthermore, \eqref{prop:qdiv-mutinf} combined with Pinsker's inequality, gives us 
\begin{proposition}\label{prop:he-and-mutinf}
	Let $\rho_{AB} \in \CD(\CH_A \otimes \CH_B)$, then $\he{\rho_{AB}}{\rho_{A} \otimes \rho_B} \le \sqrt{\Inf[\rho]{A}{B}}.$\\
\end{proposition}

Define $\Infc[\rho]{A_S}{B}{S} := \BE_{S}[\Inf[\rho]{A_S}{B}]$, then we have the following quantum version of Shearer's inequality from \cite{ATYY17}:

\begin{lemma}[Quantum Shearer's Lemma \cite{ATYY17}]\label{lemma:qshearer}
	Let $A = A_1,\ldots,A_m$ and $B$ be registers. Let $\rho \in \CD(\CH_{A} \otimes \CH_B)$ be a state such that $\rho_A = \rho_{A_1} \otimes \rho_{A_2} \otimes \cdots \otimes \rho_{A_m}$. Let $S \subseteq [m]$ be a random set independent of $\rho_{AB}$ such that $\BP[i \in S] \le \mu$ for every $i \in [m]$. Then, we have
	\[ \Infc[\rho]{A_S}{B}{S}  \le \mu \cdot \Inf[\rho]{A}{B}.\]
\end{lemma}

\subsection{Quantum Communication Complexity}

We consider quantum protocols where Alice and Bob are allowed to exchange qubits and they share some pure entangled state in the beginning, for instance a number of EPR-pairs that they are not charged for. Any lower bound in this model also translates to a lower bound in other models of quantum communication (Yao's model~\cite{yao:qcircuit} with qubit communication without prior entanglement or the Cleve-Buhrman model~\cite{cleve&buhrman:subs} with classical communication and prior entanglement). 

The total state of a quantum protocol consists of: Alice and Bob's input registers $X$ and $Y$, Alice's private register $A$,
the communication channel $C$, and Bob's private register $B$. We assume that initially Alice and Bob share some pure entangled state $\psi_{A'B'}$ where $A'$ and $B'$ are part of Alice's and Bob's private registers $A$ and $B$ respectively, while the rest of the qubits in their private workspaces are initially zero ($\ket{0}$). The channel is also initially zero. Before the start of the protocol Alice and Bob copy their inputs from the input registers to their private workspaces. Let $\ket[AB]{\psi}$ denote the state of registers $A$ and $B$ at the start. This includes the initial entangled state on $A'$ and $B'$, a bunch of zero qubits and copy of their inputs $x$ and $y$. 

Given an input distribution $p(X,Y)$ on the inputs, the starting state of the protocol is then 
$$ \pstatealt{\rho}{0}{XYABC} = \sum_{xy} p(xy) \ketbra[XY]{xy}{xy} \otimes \ketbra[AB]{\psi}{\psi} \otimes \ketbra[C]{0}{0}.$$ 

Note that the marginal state $\pstatealt{\rho}{0}{XY}$ is a classical state but not necessarily pure if $X$ and $Y$ are not independent. To make the above state a pure state, we will add purifying registers $\XR$ and $\YR$ and consider the canonical  purification of $\pstatealt{\rho}{0}{XY}$ which is the pure state 
$$\ket[X\XR Y\YR]{\rho^{(0)}} = \sum_{xy} \sqrt{p(xy)} \ket[X\XR Y\YR]{xxyy}.$$

With the above purifying registers, the initial global state of the protocol is the pure state
$$\ket[X\XR Y\YR ABC]{\rho^{(0)}}  = \ket[X\XR Y\YR]{\rho^{(0)}} \otimes \ket[AB]{\psi} \otimes \ket[C]{0}.$$ 

At each step of the protocol, either Alice or Bob applies a unitary to a subset of the registers. We will assume that they alternate: on odd rounds Alice acts and on even rounds Bob acts. We will also assume that the channel consists of one qubit. These assumptions can be made without loss of generality as they only affect the communication by a constant factor.

In an odd round $r$, Alice applies a fixed unitary transformation $\pstatealt{U}{r}{XAC} = \sum_x \ketbra[X]{x}{x}\otimes \pstatesup{U}{r}{x}{AC}$ to her private register and the channel. This corresponds to her private computation as well as to putting a one-qubit message on the channel. Note that the unitary uses the input register only as a control and does not change its contents. In an even round, Bob proceeds similarly. Hence the content of the input registers $X$ and $Y$ as well as the corresponding purifying registers $\XR$ and $\YR$ remain unchanged throughout the protocol.
 
We assume that in the last round of the protocol Bob talks. The final state of an $\ell$-round protocol (for even $\ell$) on input distribution $p(X,Y)$ is the following pure state:
$$ \ket[X\XR Y\YR ABC]{\rho^{(\ell)}} = \pstatealt{U}{\ell}{YBC} \pstatealt{U}{\ell-1}{XAC} \cdots \pstatealt{U}{1}{XAC} \ket[X\XR Y\YR ABC]{\rho^{(0)}}.$$

For technical reasons it will be convenient to assume that at the end  of the protocol, the channel contains the answer. A measurement of the channel qubit in the computational basis then determines the output bit of the protocol. We say that the protocol computes $f(x,y)$  on a distribution $p(X,Y)$ if the probability of error on the input distribution $p(X,Y)$ is at most $\eps$.  Note that we may consider the run of the protocol on a fixed input $x,y$ by taking the initial distribution $p(X,Y)$ such that $p(x,y)=1$. We say that the protocol computes $f(x,y)$ with error $\eps$ if for every input $x,y$ the probability of error is at most $\eps$.

For notational convenience, throughout this work we will sometimes write $\rho^{(r)}$ instead of $\pstate{\rho}{r}$ to denote the global state of the protocol on all the registers after round $r$. When referring to the marginal states, however, we will always write the corresponding registers.

\subsubsection*{Basic Properties of Quantum Protocols}

In the following preliminary lemmas $\pstate{\rho}{r}$ and $\pstate{\sigma}{r}$  are the states of a quantum protocol after $r$ rounds when it is run on input distributions $p(XY)$ and $q(XY)$ respectively.  Moreover, $\ell$ will denote the last round of the protocol. The following proposition is easily seen to be true since the protocol applies the same sequence of unitaries on every input $x,y$: 

\begin{proposition} \label{prop:qprotocol} There are pure states $\{\purifycond{\psi}{r}{xy}{ABC}\}_{xy}$ such that
\begin{align*}
\ \purify{\rho}{r}{\allreg} = \sum_{xy} \sqrt{p(xy)}\ket[X\XR Y\YR]{xxyy}\otimes \purifycond{\psi}{r}{xy}{ABC}\\
\ \purify{\sigma}{r}{\allreg} = \sum_{xy} \sqrt{q(xy)}\ket[X\XR Y\YR]{xxyy}\otimes \purifycond{\psi}{r}{xy}{ABC}
\end{align*}
\end{proposition}
Note that after the first round the states $\purifycond{\psi}{1}{xy}{ABC}$ only depend on $x$.

The above proposition implies that if $p(X,Y)$ is a product distribution on $X$ and $Y$, and if in a round $r$, Bob applies a unitary $\pstatealt{U}{r}{YBC}$, then the marginal states
 \[\pstatealt{\rho}{r}{X\XR Y\YR BC} = \pstatealt{U}{r}{YBC} \pstatealt{\rho}{r-1}{X\XR Y\YR BC} {\left(\pstatealt{U}{r}{YBC}\right)}^\dag \text{ and } \pstatealt{\rho}{r}{X\XR Y\YR A} = \pstatealt{\rho}{r-1}{X\XR Y\YR A} ,\]
 and a similar statement also holds when Alice acts. 

The following lemma follows easily from Proposition \ref{prop:qprotocol}:
\begin{lemma} \label{prop:qindistinguish} $\frac12\norm{\pstatealt{\rho}{r}{C}- \pstatealt{\sigma}{r}{C}} \le \stat{p(XY)-q(XY)}$.
\end{lemma}
\begin{proof} Let $\delta = \stat{p(XY)-q(XY)}$. Then using Proposition \ref{prop:qprotocol}, we can write
\begin{align*}
\ \frac12\norm{\pstatealt{\rho}{r}{C}- \pstatealt{\sigma}{r}{C}} = \frac12 \norm{\sum_{xy}(p(xy)-q(xy))\pstatesup{\psi}{r}{xy}{C}} \le \frac12 \sum_{xy}|p(xy)-q(xy)|\norm{\pstatesup{\psi}{r}{xy}{C}} \le \delta,
\end{align*}
where the second inequality is the triangle inequality and the last one follows from Proposition \ref{prop:statdef} and the fact that $\{\pstatesup{\psi}{r}{xy}{C}\}_{xy}$ are density operators and have unit trace.
\end{proof}
Using Proposition \ref{prop:tracepovm}, the above also implies that if $p(XY)$ and $q(XY)$ are $\delta$-close, then the output distributions of the protocol for both cases are $\delta$-close.

\begin{lemma}[Errors and Trace Norm]\label{prop:qerror} Given a boolean function $f(x,y)$, let $p(X,Y)$ be a distribution supported on its 0-inputs and $q(X,Y)$ be a distribution supported on its 1-inputs. If an $\ell$-round quantum protocol computes $f(x,y)$ with error $\delta$, then $\frac12\norm{\pstatealt{\rho}{\ell}{C}- \pstatealt{\sigma}{\ell}{C}} \ge 1 - 2\delta$.
\end{lemma}
\begin{proof} Recall that the last bit of the channel contains the answer and since the output of a protocol is given by a measurement of the channel qubit in the computational basis, the probabilities that the output is $0$ under $\rho_C$ and $\sigma_C$ are respectively given by $\Tr(\ketbra[C]{0}{0}\rho_C) \ge 1 - \delta$ and $\Tr(\ketbra[C]{0}{0}\sigma_C) \le \delta$.
	Using Proposition \ref{prop:tracepovm}, we have 
	\[\frac12\norm{\pstatealt{\rho}{\ell}{C}- \pstatealt{\sigma}{\ell}{C}} \ge \Tr(\ketbra[C]{0}{0}(\rho_C - \sigma_C)) \ge (1 - \delta) - \delta = 1 - 2\delta. \qedhere\]
\end{proof}

Quantum protocols have no notion of a transcript, but the following lemma still gives a bound on how much information is revealed by a quantum protocol in terms of the communication.

\begin{lemma}[Information Cost] Let $p(XY)$ be a product input distribution on $X$ and $Y$. Then, for any round $r$ in the communication protocol, it holds that
	\begin{align*}
		\ \Inf[\rho^{(r)}]{X}{Y\YR BC} \le 2r \text{ and } \Inf[\rho^{(r)}]{Y}{X\XR AC} \le 2r. 
	\end{align*}
\end{lemma}
\begin{proof}
	The proof is by induction on the number of rounds. We will only prove the first inequality as the second one follows analogously. When $r=0$, no messages have been exchanged and since $p(x,y)=p(x)p(y)$ for any $x,y$, it follows that the initial state is of the form $\pstate{\rho}{0} = \pstatealt{\rho}{0}{X\XR} \otimes \pstatealt{\rho}{0}{Y\YR} \otimes \pstatealt{\rho}{0}{ABC}$. So, using Proposition \ref{prop:qind}, it follows that $\Inf[\rho^{(0)}]{X}{Y\YR BC} = 0.$
	
	Now, let us assume that the statement holds for $r-1$ rounds. When $r$ is even, Bob applies a unitary $\pstatealt{U}{r}{YBC}$. Since $p(XY)$ is a product distribution on $X$ and $Y$, from Proposition \ref{prop:qprotocol}, it follows that $\pstatealt{\rho}{r}{XY\YR BC} = \pstatealt{U}{r}{YBC} \pstatealt{\rho}{r-1}{XY\YR BC}{\left(\pstatealt{U}{r}{YBC}\right)}^\dag$. Hence, using Proposition \ref{prop:local}, we have
	\begin{align*}
	\ \Inf[\rho^{(r)}]{X}{Y\YR BC} &=\Inf[\rho^{(r-1)}]{X}{Y\YR BC} \le 2(r-1),
	\end{align*}
where the inequality follows from the inductive hypothesis.
	
	When $r$ is odd, Alice applies a unitary $\pstatealt{U}{r}{XAC}$ with $X$ as control. Using chain rule, we can write
	\begin{align*}
	\ \Inf[\rho^{(r)}]{X}{Y\YR BC} &= \Inf[\rho^{(r)}]{X}{Y\YR B} + \Infc[\rho^{(r)}]{X}{C}{Y\YR B} \\
	\							  &\le 	\Inf[\rho^{(r)}]{X}{Y\YR B} + 2	= \Inf[\rho^{(r-1)}]{X}{Y\YR B} + 2 \\
	\							  & \le  \Inf[\rho^{(r-1)}]{X}{Y\YR BC} + 2 \le 2(r-1) + 2 = 2r,
	\end{align*}
where the first inequality follows from Proposition \ref{prop:qmutinfbound}, the second equality follows since $\rho^{(r)}_{XY\YR B} = \rho^{(r-1)}_{XY\YR B}$ as Alice applies a unitary $U_{XAC}$ with $X$ as a control register, and the second inequality follows from chain rule and non-negativity of conditional mutual information.
\end{proof}

\section{Classical Communication Lower Bound}
\label{sec:classical}

In this section, we present a new proof of the classical communication lower bound that we will later generalize to the quantum setting. We will prove that any randomized protocol for the $\sink$ function that errs with probability at most $1/3$ must communicate at least $\Omega(t)$ bits. 

As is standard, to prove this we use a hard distribution $p(XY)$ on the inputs.  

\paragraph{Hard Input Distribution $p(X,Y)$:} Let $p_0(X,Y)$ and $p_1(X,Y)$ denote the uniform distribution on $\sink^{-1}(0)$ and $\sink^{-1}(1)$ respectively. In the input distribution $p(X,Y)$, the input is sampled from $p_0(X,Y)$ with probability $\frac12$ and from $p_1(X,Y)$ with probability $\frac12$.  \\  

Since we have a distribution on the inputs, we may assume without loss of generality that the randomized protocol is deterministic.
We will prove a lower bound on the communication by showing that if the length of the messages of the protocol is at most $\frac12\eps^3 t$, then the distribution of the messages looks almost the same under the distributions $p_0(X,Y)$ and $p_1(X,Y)$: denoting by $p_0(M)$ and $p_1(M)$ the induced distributions on the messages under $p_0(X,Y)$ and $p_1(X,Y)$, respectively, we will show that $p_0(M)$ and $p_1(M)$ are  $8\eps$-close in statistical distance. To show this, we use the \emph{fooling distribution} method from \cite{RS15}. We will give another distribution $u(X,Y)$ such that the induced distribution $u(M)$ will be $4\eps$-close to each of $p_0(M)$ and $p_1(M)$.  For the $\sink$ function, this
{\bf fooling distribution} $u(X,Y)$ is the uniform distribution on $\bits^{n+n}$.
More precisely, we prove:

\begin{theorem}\label{thm:statclose} 
Let $\eps > 0$ be a constant and $t$ be large enough. Then, for any deterministic protocol for the $\sink$ function with communication at most $\frac12\eps^3 t$, we have that 
$p_0(M) \papprox{4\eps} u(M) \papprox{4\eps} p_1(M).$\\
\end{theorem}

Since the input distribution $p(X,Y)$ is balanced, using Lemma \ref{lemma:error}, the distributions $p_0(M)$ and $p_1(M)$ must have statistical distance at least $1/3$ if the protocol has error $1/3$ on $p(X,Y)$. So, it must be that $8\eps \ge 1/3$, and hence $\eps \ge 1/24$, and the $\Omega(t)$ lower bound on the communication (Theorem~\ref{th:CMS}) follows. 

Next, we prove Theorem \ref{thm:statclose}. Before the proof, it will be helpful to keep in mind how the distributions $p_1(X,Y)$, $p_0(X,Y)$ and $u(X,Y)$ are related. Note that by definition, $p_0(X,Y) = u(X,Y | \sink = 0)$ and $p_1(X,Y) = u(X,Y| \sink = 1)$. Also, notice that the input distributions $p_0(X,Y)$ and $u(X,Y)$ are already very close in statistical distance:

\begin{claim}\label{claim:zerodist}
$p_0(X,Y) \papprox{\gamma} u(X,Y)$ with $\gamma = t2^{-(t-1)} = o(1)$.
\end{claim}
\begin{proof}
Note that under the uniform distribution $u(XY)$, the probability that the function $\sink$ takes value $1$ is exactly $t2^{-(t-1)}$,  because for each vertex $v$, the event that $v$ is the sink has probability exactly $2^{-(t-1)}$, and these events are disjoint for the $t$ vertices. This means that \[u(\sink = 0) = 1 - t2^{-(t-1)}.\]
	 Since $p_0(XY) = u(XY|\sink = 0)$,  Lemma~\ref{prop:statcond} implies $\stat{p_0(XY) - u(XY)} \le t2^{-(t-1)} = \gamma$.
\end{proof}

Furthermore, recall that we can generate the distribution $p_1(X,Y)$ from $u(X,Y)$ by conditioning on a simple collision event: for any vertex $v$, denoting by $N(v)$ the set of edges incident on $v$, the distribution $p_1(X,Y)$ can be generated from $u(X,Y)$ by first picking a uniformly random vertex $V \in [t]$ as the sink, and then conditioning on the event that $X_{N(V)}=Y_{N(V)} \oplus z_{N(V)}$, where $z_{N(v)}$ is the unique string that encodes the orientations of the edges in $N(v)$ when vertex $v$ is the sink.

To complete the proof, we use the following lemma from \cite{RS15}, which bounds the effect of conditioning on a collision event (for completeness we include a proof in Appendix~\ref{appA}). 

\begin{restatable}[Lemma 4.3 in \cite{RS15}]{lemma}{condition} \label{lemma:condition} {Given a probability space $q$, if $A,B \in [r]$ are uniform and independent random variables, and $A-C-B$,} then
	\[q(C) \papprox{\epsilon} q(C|A=B), \text{ with } \epsilon = 2\sqrt[3]{\Inf[q]{C}{A}}+ 2\sqrt[3]{\Inf[q]{C}{B}}.\]
\end{restatable}

\begin{proof}[Proof of Theorem \ref{thm:statclose}]
	Since $\stat{p_0(XY)-u(XY)} \le t2^{-(t-1)} = o(1)$ from Claim \ref{claim:zerodist}, this already implies that $\stat{p_0(M)-u(M)}=o(1)$ since $M$ is a function of $X,Y$. So, we focus on bounding $\stat{p_1(M) - u(M)}$. For this, let $V \in [t]$ and let $u(V)$ be the uniform distribution on $[t]$. Recall that 
	$$
	p_1(XY) = \BE_{u(v)}[u(XY|X_{N(v)}=Y_{N(v)} \oplus z_{N(v)})].
	$$
	We will show that under the fooling distribution $u(XY)$, the messages of the protocol contain little information about $X_{N(V)}$ and $Y_{N(V)}$. Lemma~\ref{lemma:condition} and concavity will then complete the proof.

	Note that for any fixed edge $e$, it holds that $u(e \in N(V)) = \frac2t$. Since under $u(XY)$, the binary random variables $X_e$ (resp.\ $Y_e$) and $X_{e'}$ (resp.\ $Y_{e'}$) are mutually independent for any two edges $e$ and $e'$, applying Shearer's inequality (Lemma \ref{lemma:shearer}), we get that
	\begin{align}\label{eqn:shearer}
		\Infc[u]{X_{N(V)}}{M}{V} & \le  \frac2t \cdot \Inf[u]{X}{M} \le \frac2t \cdot |M| \le \eps^3 \text{ and }\notag\\   \Infc[u]{Y_{N(V)}}{M}{V} & \le  \frac2t \cdot \Inf[u]{Y}{M} \le \frac2t \cdot |M| \le \eps^3.
	\end{align}
	Note that for any $v$, shifting $Y_{N(v)}$ by a fixed string $z_{N(v)}$ does not change the mutual information $\Inf[u]{Y_{N(v)}}{M}$. Furthermore, since $X$ and $Y$ are independent $X - M - Y$ holds. Hence, using Proposition \ref{prop:convexstat} and Lemma \ref{lemma:condition} (with $A=X_{N(v)}$, $B=Y_{N(v)}\oplus z_{N(v)}$, $C=M$), it holds that
	\begin{align*}
	\ \stat{p_1(M) - u(M)} &= \stat{\BE_{u(v)}[u(M|X_{N(v)}=Y_{N(v)} \oplus z_{N(v)})] - u (M)} \\
	\					&= \BE_{u(v)}\left[\stat{~u(M|X_{N(v)}=Y_{N(v)}\oplus z_{N(v)}) - u(M)~}\right] \\
	\					&\le 2\BE_{u(v)} \sqrt[3]{\Inf[u]{X_{N(v)}}{M}} + 2\BE_{u(v)} \sqrt[3]{\Inf[u]{Y_{N(v)}}{M}}.
	\end{align*}
	Further, using concavity of the cube root function over non-negative reals and \eqref{eqn:shearer}, we get that
	\begin{align*}
	\	\stat{p_1(M) - u(M)} &\le 2 \sqrt[3]{\BE_{u(v)} \Inf[u]{X_{N(v)}}{M}} + 2 \sqrt[3]{\BE_{u(v)} \Inf[u]{Y_{N(v)}}{M}}\\
	\					  &= 2 \sqrt[3]{\Infc[u]{X_{N(V)}}{M}{V}} + 2\sqrt[3]{\Infc[u]{X_{N(V)}}{M}{V}} \stackrel{\eqref{eqn:shearer}}{\le} 4\eps.
	\end{align*}
	This shows that for $t$ large enough, $\stat{p_0(M)-u(M)} \le 4 \eps$ and $\stat{p_1(M)-u(M)}\le 4\eps$, concluding the proof.
\end{proof}

\section{Quantum Communication Lower Bound}
\label{sec:quantum}

The proof for the quantum case proceeds similarly to the classical case with some minor differences. Let $p_0(XY)$, and $u(XY)$ be as before: $p_0(XY)$ is uniform on $\sink^{-1}(0)$ and $u(XY)$ is the uniform distribution. Fix an $\ell$-qubit protocol where per our convention $\ell$ is even as Bob sends the last message. Let $\bstate{\omicron}{\ell}$ and $\bstate{\upsilon}{\ell}$ be the final pure states of the protocol on distributions $p_0(XY)$ and $u(XY)$, respectively. Let $V \in [t]$, let $u(V)$ denote the uniform distribution on $[t]$ and let $\iota^{v,(\ell)}$ denote the final pure state of the protocol when run on distribution $u(XY|X_{N(v)}=Y_{N(v)} \oplus z_{N(v)})$, that is, when vertex~$v$ is the sink. Note that the distribution  $u(XY|X_{N(v)}=Y_{N(v)} \oplus z_{N(v)})$ is supported on only the $1$-inputs to the $\sink$ function.  If the protocol computes the $\sink$ function with error at most $1/3$ on every input, then Lemma~\ref{prop:qerror} implies 
\begin{equation}\label{eq:finalstatesclose}
\BE_{u(v)}\left[\norm{\fixedvstate{\ell}{C}- \pstatealt{\omicron}{\ell}{C}}\right]\ge 2/3.
\end{equation}
We are going to argue that if $\ell \ll t^{1/3}$, then the distribution $u(XY)$ is also a fooling distribution for quantum protocols. That is, it must be the case that both $\pstatealt{\omicron}{\ell}{C} \approx \pstatealt{\upsilon}{\ell}{C}$ and, for a typical vertex~$v$, $\fixedvstate{\ell}{C}\approx \pstatealt{\upsilon}{\ell}{C}$
(and hence $\fixedvstate{\ell}{C} \approx  \pstatealt{\omicron}{\ell}{C}$ for a typical $v$).

\begin{theorem}\label{thm:traceclose} 
Let $\eps > 0$ be a constant and $t$ be large enough. Then, for any quantum protocol for the $\sink$ function with communication complexity at most $\ell = \frac18 \eps^{2/3} t^{1/3}$, we have that 
$$
\BE_{u(v)}\left[\norm{\fixedvstate{\ell}{C}- \pstatealt{\omicron}{\ell}{C}}\right]\le \eps.
$$
\end{theorem}

Combining this theorem with \eqref{eq:finalstatesclose} immediately implies the quantum communication lower bound of $\Omega(t^{1/3})$ promised by Theorem~\ref{th:qlogrank}.

First, $\pstatealt{\omicron}{\ell}{C} \approx \pstatealt{\upsilon}{\ell}{C}$ is clear because $p_0(X,Y)\approx u(X,Y)$ (see Lemma~\ref{prop:qindistinguish}). To prove that $\fixedvstate{\ell}{C}\approx\pstatealt{\upsilon}{\ell}{C}$ for a typical $v$, we will use Lemma~3.6 from \cite{ATYY17} (we state it a bit differently here to make it easier for our application). This allows us to relate the fooling distribution with the input distribution similar to the role of Lemma~\ref{lemma:condition} in the proof for the classical case. The proof of this lemma is an involved round-by-round induction; we include a proof in Appendix~\ref{appB} for completeness. 

\begin{restatable}[Lemma 3.6 in \cite{ATYY17}]{lemma}{qcondition} \label{lemma:qcond}
	Let $X=X_1X_2$ and $Y=Y_1Y_2$ be random variables where $X,Y \in \bits^n$. Let $u'(XY)$ be the uniform distribution on $XY$ and let $q(XY)=u'(XY|X_1=Y_1)$ be another distribution. For every $s \le r$, let $\bstate{\rho}{s}$ and $\bstate{\sigma}{s}$ denote the state of a quantum protocol after $s$ rounds on distributions $u'(XY)$ and $q(XY)$ respectively. If for every $s \le r$,  we have
	\[ \Inf[\rho^{(s)}]{X_1}{Y\YR BC} \le \eps_s \text{ for odd } s, \text{ and } \Inf[\rho^{(s)}]{Y_1}{X\XR AC} \le \eps_s \text{ for even } s,\]
then it holds that
$\norm{\pstatealt{\sigma}{r}{X_1Y_1 C} - \pstatealt{\sigma}{r}{X_1Y_1} \otimes \pstatealt{\rho}{r}{C}}  \le 4\sqrt{2}\sum_{s=1}^r \sqrt{\eps_s}.$\\
	
\end{restatable}

Fix a vertex $v\in[t]$. Define $\eps_{v,s} = \Inf[\upsilon^{(s)}]{X_{N(v)}}{Y\YR BC}$ for odd rounds~$s$, and $\eps_{v,s}=\Inf[\upsilon^{(s)}]{Y_{N(v)}}{X\XR AC}$ for even rounds~$s$.
If these $\eps_{v,s}$'s are mostly small, then $\fixedvstate{\ell}{C} \approx  \pstatealt{\upsilon}{\ell}{C}$:

\begin{lemma}\label{claim:qfooling} 
		$\norm{\iota^{v,(\ell)}_{C} - \pstatealt{\upsilon}{\ell}{C}} \le 4\sqrt{2}\sum_{s=1}^{\ell} \sqrt{\eps_{v,s}}.$
\end{lemma}
\begin{proof}
To apply Lemma \ref{lemma:qcond}, we will choose $X_1=X_{N(v)}$, $X_2=X_{N(v)^c}$ and $Y_1=Y_{N(v)}\oplus z_{N(v)}, Y_2 = Y_{N(v)^c}$ and $u'(XY) = u(XY)$. Note that $u'(XY)$ is still the uniform distribution. Furthermore, using Proposition \ref{prop:qprotocol}, for every $s$, the state $\rho^{(s)}$ in Lemma \ref{lemma:qcond} is the same as $\mu^{(s)}$  after a suitable relabeling. Hence, it follows that $\Inf[\rho^{(s)}]{X_1}{Y\YR BC} = \Inf[\upsilon^{(s)}]{X_{N(v)}}{Y\YR BC} = \eps_{v,s}$ for odd $s$, and $\Inf[\rho^{(s)}]{Y_1}{X\XR AC} = \Inf[\upsilon^{(s)}]{Y_{N(v)}}{X\XR AC} = \eps_{v,s}$ for even $s$. 

Now, we apply Lemma \ref{lemma:qcond}. Since $\Tr_{X_1Y_1}\left(\iota^{v,(\ell)}_{X_1Y_1C}\right) = \iota^{v,(\ell)}_{C}$ and $\Tr_{X_1Y_1}\left({\iota}^v_{X_1Y_1} \otimes \pstatealt{\upsilon}{\ell}{C}\right) = \pstatealt{\upsilon}{\ell}{C}$, we get that  
$\norm{\iota^{v,(\ell)}_{C} - \pstatealt{\upsilon}{\ell}{C}} \le \norm{\iota^{v,(\ell)}_{X_1Y_1C} - {\iota}^v_{X_1Y_1} \otimes \pstatealt{\upsilon}{\ell}{C}} \le 4\sqrt{2}\sum_{s=1}^{\ell} \sqrt{\eps_{v,s}}$.
\end{proof}

We move on to the proof of the theorem now.

\begin{proof}[Proof of Theorem \ref{thm:traceclose}] 
	Recall that $\stat{p_0(XY)-u(XY)} \le t2^{-(t-1)} = o(1)$ from Claim \ref{claim:zerodist}, and using Lemma \ref{prop:qindistinguish}, this already implies that $\norm{\pstatealt{\omicron}{\ell}{C} - \pstatealt{\upsilon}{\ell}{C}} = o(1)$.
	
	Let us turn to bounding $\BE_{u(v)}\left[\norm{\fixedvstate{\ell}{C}- \pstatealt{\upsilon}{\ell}{C}}\right]$. We first show that under the fooling distribution $u(XY)$, the states of the quantum protocol contain little information about $X_{N(V)}$ and $Y_{N(V)}$. Then, applying Claim \ref{claim:qfooling} and appealing to concavity will complete the proof similar to the classical case.

	Note that for any fixed edge $e$, it holds that $u(e \in N(V)) = \frac2t$, and also recall that under $u(XY)$, the random variables $X_e$ (resp.\ $Y_e$) and $X_{e'}$ (resp.\ $Y_{e'}$) are mutually independent for any two edges $e$ and $e'$. Therefore, using Proposition~\ref{prop:qprotocol} the state $\pstatealt{\upsilon}{\ell}{X} = \tens{e} \pstatealt{\upsilon}{\ell}{X_e}$ (and similarly $\pstatealt{\upsilon}{\ell}{Y} = \tens{e} \pstatealt{\upsilon}{\ell}{Y_e}$). Hence, applying the quantum version of Shearer's inequality (Lemma~\ref{lemma:qshearer}) and using Proposition~\ref{prop:qmutinfbound}, for every round $s \le \ell$ we get that
	\begin{align}\label{eqn:qshearer}
		 	 \Infc[\upsilon^{(s)}]{X_{N(V)}}{Y\YR BC}{V} &\le \frac2t \cdot  \Inf[\upsilon^{(s)}]{X}{Y\YR BC} \le \frac{4\ell}{t} \text{ and} \notag\\
		\Infc[\upsilon^{(s)}]{Y_{N(V)}}{X\XR AC}{V} &\le \frac2t \cdot \Inf[\upsilon^{(s)}]{Y}{X\XR AC} \le \frac{4\ell}{t}.
	\end{align}
Further using Lemma~\ref{claim:qfooling}, concavity, and \eqref{eqn:qshearer} we get 
\begin{align*}
	\  \BE_{u(v)}\left[\norm{\fixedvstate{\ell}{C}- \pstatealt{\upsilon}{\ell}{C}}\right] &\le 4\sqrt{2} ~\BE_{u(v)}\left[\sum_{s=1}^{\ell} \sqrt{\eps_{v,s}}\right] \\
	\		&\le 4\sqrt{2} ~\sum_{s=1}^{\ell} \sqrt{\BE_{u(v)}[\eps_{v,s}]} \\
	\		&=4\sqrt{2} \sum_{\substack{s=1\\s \text{ odd}}}^{\ell} \sqrt{\Infc[\upsilon^{(s)}]{X_{N(V)}}{Y\YR BC}{V}}  + 4 \sqrt{2}\sum_{\substack{s=1\\s \text{ even}}}^{\ell}  \sqrt{\Infc[\upsilon^{(s)}]{Y_{N(V)}}{X\XR AC}{V}}\\
	\ &\le4\sqrt{2}\cdot\frac\ell{2} \sqrt\frac{4\ell}{t} + 4\sqrt{2}\cdot \frac\ell{2}\sqrt\frac{4\ell}{t}= \sqrt\frac{128\ell^3}{t} \le \frac\eps2.
	\end{align*}
Using the triangle inequality, we get that for large enough $t$, the following holds
	\[ \BE_{u(v)}\left[\norm{\fixedvstate{\ell}{C}- \pstatealt{\omicron}{\ell}{C}}\right] \le \BE_{u(v)}\left[\norm{\fixedvstate{\ell}{C}- \pstatealt{\upsilon}{\ell}{C}}\right] + \norm{\pstatealt{\omicron}{\ell}{C} - \pstatealt{\upsilon}{\ell}{C}} \le \frac\eps2 + o(1) \le \eps.\qedhere
	\]
\end{proof}

\section{Future Work}

One obvious remaining open problem is to close the gap between the current lower bound of $\Omega(t^{1/3})$ on the quantum communication complexity of the $\sink$ function, and the best known upper bound of $\widetilde{O}(\sqrt{t})$.
We conjecture the upper bound is essentially tight.
One way to improve our lower bound would be to improve Lemma~\ref{claim:qfooling}, maybe with a different distance measure. 

The main question left open by this work, as well as by~\cite{CMS18,AGT18}, is of course 
the status of the (non-approximate) log-rank conjecture itself.
The proof that the $\sink$ function has low approximate rank crucially uses the fact that the identity matrix has low approximate rank (which follows from the fact that the equality function has low randomized communication complexity). In contrast, the actual (non-approximate) rank of the identity matrix is as large as its dimension. Accordingly, it is not so clear what examples like the $\sink$ function suggest for the status of the log-rank conjecture itself. We are not sure what to conjecture about that conjecture.

\subsection*{Acknowledgments.} 
Thanks to Sander Gribling for useful discussions and a heartfelt gratitude to him for careful proofreading of our first draft. Thanks to Siva Ramamoorthy for helpful comments. We also thank Arkadev Chattopadhyay, Nikhil Mande and Suhail Sherif for answering questions about their paper~\cite{CMS18}, as well as Anurag Anshu, Naresh Goud Boddu and Dave Touchette for helpful correspondence about their independent work~\cite{AGT18} and helpful comments on our draft.
RdW thanks Srinivasan Arunachalam for pointing him to~\cite{CMS18} when it had just appeared on ECCC, and for helpful comments on the draft.

\bibliographystyle{alpha}
{\small
\bibliography{logrank}

\newcommand{\etalchar}[1]{$^{#1}$}
\begin{thebibliography}{ABDG{\etalchar{+}}17}

\bibitem[AA05]{aaronson&ambainis:searchj}
Scott Aaronson and Andris Ambainis.
\newblock Quantum search of spatial regions.
\newblock {\em Theory of Computing}, 1(1):47--79, 2005.
\newblock Earlier version in FOCS'03. quant-ph/0303041.

\bibitem[ABBD{\etalchar{+}}16]{anshuetal:sepcheatsheets}
Anurag Anshu, Aleksandrs Belovs, Shalev Ben-David, Mika G{\"{o}}{\"{o}}s, Rahul
  Jain, Robin Kothari, Troy Lee, and Miklos Santha.
\newblock Separations in communication complexity using cheat sheets and
  information complexity.
\newblock In {\em Proceedings of 57th IEEE FOCS}, pages 555--564, 2016.
\newblock arXiv:1605.01142.

\bibitem[ABDG{\etalchar{+}}17]{anshuetal:sepapproxrank}
Anurag Anshu, Shalev Ben-David, Ankit Garg, Rahul Jain, Robin Kothari, and Troy
  Lee.
\newblock Separating quantum communication and approximate rank.
\newblock In {\em Proceedings of Computational Complexity Conference (CCC'17)},
  pages 24:1--24:33, 2017.
\newblock arXiv:1611.05754.

\bibitem[ABT18]{AGT18}
Aurag Anshu, {Naresh Goud} Boddu, and Dave Touchette.
\newblock Quantum {L}og-{A}pproximate-{R}ank conjecture is also false, November
  2018.
\newblock Private communication, to appear.

\bibitem[ATYY17]{ATYY17}
Anurag Anshu, Dave Touchette, Penghui Yao, and Nengkun Yu.
\newblock Exponential separation of quantum communication and classical
  information.
\newblock In {\em Proceedings of 49th ACM STOC}, pages 277--288, 2017.

\bibitem[BCW98]{BuhrmanCleveWigderson98}
Harry Buhrman, Richard Cleve, and Avi Wigderson.
\newblock Quantum vs.~classical communication and computation.
\newblock In {\em Proceedings of 30th ACM STOC}, pages 63--68, 1998.
\newblock quant-ph/9802040.

\bibitem[BW01]{buhrman&wolf:qcclower}
Harry Buhrman and Ronald~{de} Wolf.
\newblock Communication complexity lower bounds by polynomials.
\newblock In {\em Proceedings of 16th IEEE Conference on Computational
  Complexity}, pages 120--130, 2001.
\newblock cs.CC/9910010.

\bibitem[CB97]{cleve&buhrman:subs}
Richard Cleve and Harry Buhrman.
\newblock Substituting quantum entanglement for communication.
\newblock {\em Physical Review A}, 56(2):1201--1204, 1997.
\newblock quant-ph/9704026.

\bibitem[CMS18]{CMS18}
Arkadev Chattopadhyay, Nikhil~S. Mande, and Suhail Sherif.
\newblock The log-approximate-rank conjecture is false.
\newblock {\em Electronic Colloquium on Computational Complexity {(ECCC)}},
  25:176, 2018.

\bibitem[CT06]{CT06}
Thomas~M. Cover and Joy~A. Thomas.
\newblock {\em Elements of Information Theory (Wiley Series in
  Telecommunications and Signal Processing)}.
\newblock Wiley-Interscience, 2006.

\bibitem[GJPW17]{GJPW17}
Mika {G\"{o}\"{o}s}, T.~S. Jayram, Toniann Pitassi, and Thomas Watson.
\newblock Randomized communication vs.\ partition number.
\newblock In {\em Proceedings of 44th ICALP}, pages 52:1--52:15, 2017.

\bibitem[GKR16]{GKR14}
Anat Ganor, Gillat Kol, and Ran Raz.
\newblock Exponential separation of information and communication for boolean
  functions.
\newblock {\em Journal of the {ACM}}, 63(5):46:1--46:31, 2016.

\bibitem[GPW15]{GPW15}
Mika {G\"{o}\"{o}s}, Toniann Pitassi, and Thomas Watson.
\newblock Deterministic communication vs.\ partition number.
\newblock In {\em Proceedings of 56th FOCS}, pages 1077--1088, 2015.

\bibitem[Gro96]{grover:search}
Lov~K. Grover.
\newblock A fast quantum mechanical algorithm for database search.
\newblock In {\em Proceedings of 28th ACM STOC}, pages 212--219, 1996.
\newblock quant-ph/9605043.

\bibitem[HHL16]{HHL:xor}
Hamed Hatami, Kaave Hosseini, and Shachar Lovett.
\newblock Structure of protocols for {XOR} functions.
\newblock In {\em Proceedings of 57th IEEE FOCS}, pages 282--288, 2016.

\bibitem[KN97]{KN96}
Eyal Kushilevitz and Noam Nisan.
\newblock {\em Communication {C}omplexity}.
\newblock Cambridge University Press, New York, NY, USA, 1997.

\bibitem[KR11]{klartag&regev:lowerbound}
Boaz Klartag and Oded Regev.
\newblock Quantum one-way communication is exponentially stronger than
  classical communication.
\newblock In {\em Proceedings of 43rd ACM STOC}, 2011.
\newblock arXiv:1009.3640.

\bibitem[Lov14]{lovett:advancesrank}
Shachar Lovett.
\newblock Recent advances on the log-rank conjecture in communication
  complexity.
\newblock {\em Bulletin of the EATCS}, 112, 2014.

\bibitem[Lov16]{lovett:rootrank}
Shachar Lovett.
\newblock Communication is bounded by root of rank.
\newblock {\em Journal of the ACM}, 63(1):1:1--1:9, 2016.
\newblock Earlier version in STOC'14.

\bibitem[LS93]{lovasz&saks:cc}
L{\'{a}}szlo Lov{\'{a}}sz and Michael Saks.
\newblock Communication complexity and combinatorial lattice theory.
\newblock {\em Journal of Computer and System Sciences}, 47:322--349, 1993.
\newblock Earlier version in FOCS'88.

\bibitem[LS09]{lee&shraibman:survey}
Troy Lee and Adi Shraibman.
\newblock Lower bounds in communication complexity.
\newblock {\em Foundations and Trends in Theoretical Computer Science},
  3(4):263--398, 2009.

\bibitem[MS82]{mehlhorn&schmidt:lasvegas}
Kurt Mehlhorn and Erik Schmidt.
\newblock {L}as {V}egas is better than determinism in {VLSI} and distributed
  computing.
\newblock In {\em Proceedings of 14th ACM STOC}, pages 330--337, 1982.

\bibitem[Raz99]{raz:qcc}
Ran Raz.
\newblock Exponential separation of quantum and classical communication
  complexity.
\newblock In {\em Proceedings of 31st ACM STOC}, pages 358--367, 1999.

\bibitem[Rot14]{R14}
Thomas Rothvo{\ss}.
\newblock A direct proof for {L}ovett's bound on the communication complexity
  of low rank matrices.
\newblock {\em CoRR}, abs/1409.6366, 2014.

\bibitem[RS15]{RS15}
Anup Rao and Makrand Sinha.
\newblock Simplified {S}eparation of {I}nformation and {C}ommunication.
\newblock {\em Electronic Colloquium on Computational Complexity {(ECCC)}},
  22:57, 2015.

\bibitem[RY18]{RY18}
Anup Rao and Amir Yehudayoff.
\newblock {\em Communication {C}omplexity}.
\newblock Textbook (Draft), 2018.

\bibitem[TWXZ13]{tsang:logrank}
Hing~Yin Tsang, Chung~Hoi Wong, Ning Xie, and Shengyu Zhang:.
\newblock Fourier sparsity, spectral norm, and the log-rank conjecture.
\newblock In {\em Proceedings of 54th IEEE FOCS}, pages 658--667, 2013.
\newblock arXiv:1304.1245.

\bibitem[Wat18]{W18}
John Watrous.
\newblock {\em The {T}heory of {Q}uantum {I}nformation}.
\newblock Cambridge University Press, 2018.

\bibitem[Wil13]{W13}
Mark~M. Wilde.
\newblock {\em Quantum {I}nformation {T}heory}.
\newblock Cambridge University Press, 2013.

\bibitem[Yao93]{yao:qcircuit}
Andrew C-C. Yao.
\newblock Quantum circuit complexity.
\newblock In {\em Proceedings of 34th IEEE FOCS}, pages 352--360, 1993.

\bibitem[Zha14]{zhang:quantumxor}
Shengyu Zhang.
\newblock Efficient quantum protocols for {XOR} functions.
\newblock In {\em Proceedings of 25th SODA}, pages 1878--1885, 2014.
\newblock arXiv:1307.6738.

\end{thebibliography}
}

\appendix

\section{Proof of Lemma \ref{lemma:condition}}\label{appA}

\condition*

\begin{proof}
We assume $\Inf[q]{C}{A}, \Inf[q]{C}{B} \le 1$, since otherwise the lemma is trivially true. For brevity, set 
\begin{align*}
	\alpha^3 &= \Inf[q]{C}{A} = \Ex{q(c)}{\Div{q(A|c)}{q(A)}} \text{ and } \beta^3 = \Inf[q]{C}{B} = \Ex{q(c)}{\Div{q(B|c)}{q(B)}}. 
\end{align*}

Call $c$ \emph{bad} if $\Div{q(A|c)}{q(A)} \geq \alpha^{2}$ or $\Div{q(B|c)}{q(B)} \geq \beta^{2}$, and \emph{good} otherwise. By Markov's inequality, the probability that $C$ is bad is at most $\alpha+\beta$. To prove Lemma \ref{lemma:condition}, we need the following claim proved in \cite{GKR14}. For completeness, we include the short proof after finishing the proof of Lemma \ref{lemma:condition}. 

\begin{claim} \label{claim:equal} Given independent random variables $\bA,\bB \in [r]$ in  a probability space $q$, if $\bA$ is $\gamma_1$-close to uniform, and $\bB$ is $\gamma_2$-close to uniform, then $q(\bA=\bB) \ge \dfrac{1-\gamma_1-\gamma_2}{r}$.\\
\end{claim}

When $c$ is good,  Pinsker's inequality (Proposition  \ref{proposition:pinsker}) implies that conditioned on $c$, $A$ is $\alpha$-close to uniform and $B$ is $\beta$-close to uniform. Then, since $A - C - B$, using Claim \ref{claim:equal} (with $\bA=A$ and $\bB=B$ in the probability space $q$ conditioned on $c$) implies that $q(A=B|c) \ge \frac{1-\alpha-\beta}{r}$. Since $q(A=B) = \frac1r$, we have that for a good $c$,
\begin{align}
q(c|A=B) = \frac{q(c)\cdot q(A=B|c)}{q(A=B)} \ge (1- \alpha-\beta) \cdot q(c). \label{eqn:ratio}
\end{align}

For any event $\CQ$, \equationref{eqn:ratio} implies that
\begin{align*}
q(C \in \CQ) - q(C \in \CQ | A=B) &\leq \sum_{c \in \CQ, c \text{ bad}} q(c) + \sum_{c \in \CQ, c \text{ good}} (q(c) - q(c|A=B))\\
\ &\leq q(C \text{ is bad}) + \sum_{c}q(c) (\alpha+\beta) \\
\ &\le \alpha + \beta + \sum_{c}q(c) (\alpha+\beta) \leq 2\alpha+ 2\beta,
\end{align*}
and since $\stat{q(C) - q(C|A=B)} = \max_{\CQ} (q(C \in \CQ) - q(C \in \CQ | A=B))$ we get the required upper bound on statistical distance. 
\end{proof}

\begin{proof}[Proof of Claim \ref{claim:equal}]
	For each $i$, let $q(\bA=i) = \frac1r + \alpha_i$ and $q(\bB=i) = \frac1r + \beta_i$. Then, $\sum_i \alpha_i = \sum_i \beta_i = 0$, and $\alpha_i, \beta_i \geq -\frac{1}{r}$. Using these facts,
\begin{align*}
q(\bA=\bB) & =  \sum_i \left(\frac1r+\alpha_i\right)\left(\frac1r + \beta_i\right) \\
& = \frac1r + \frac{\sum_i \alpha_i}{r} + \frac{\sum_i \beta_i}{r} + \sum_i \alpha_i \beta_i = \frac1r + \sum_i \alpha_i \beta_i.
\end{align*}
To lower bound the above, we will only consider the negative terms in the summation:
\begin{align*}
q(\bA=\bB) &\geq \frac1r + \sum_{i: \alpha_i >0, \beta_i <0} \alpha_i \beta_i + \sum_{i: \alpha_i <0, \beta_i >0} \alpha_i \beta_i 
		  \geq \frac1r - \frac1r  \sum_{i: \alpha_i >0} \alpha_i  - \frac1r\sum_{i:  \beta_i >0} \beta_i. 
\end{align*}

From Proposition \ref{prop:statdef}, it follows that $\sum_{i: \alpha_i >0} \alpha_i$ is the statistical distance~$\gamma_1$ between $\bA$ and the uniform distribution on $[r]$ and likewise for $\bB$. So we get
\[ q(\bA=\bB) \ge \frac{1-\gamma_1-\gamma_2}{r}. \qedhere\]
\end{proof}

\section{Proof of Lemma \ref{lemma:qcond}}\label{appB}

\newcommand{\boldX}{\mathbf{X}_1}
\newcommand{\pX}{X'_1}
\newcommand{\pY}{Y'_1}
\newcommand{\py}{y'_1}
\newcommand{\px}{x'_1}
\newcommand{\mainreg}{\pX\pY}

\renewcommand{\rest}{X_2 \overline{X_2} Y_2 \overline{Y_2} AB}
\renewcommand{\BX}{\mathbf{X'_1}}
\renewcommand{\bX}{\mathbf{x'_1}}

\renewcommand{\BY}{\mathbf{Y'_1}}
\renewcommand{\bY}{\mathbf{y'_1}}
\renewcommand{\alt}{\BX \BY}
\newcommand{\bx}{\bX}
\newcommand{\by}{\bY}

\makeatletter
\newcommand{\leqnos}{\tagsleft@true\let\veqno\@@leqno}
\newcommand{\reqnos}{\tagsleft@false\let\veqno\@@eqno}
\reqnos
\makeatother

\qcondition*

To simplify the notation in the proof, define $R = \rest$, and $\pX = X_1 \overline{X_1}$, $\pY = Y_1 \overline{Y_1}$, $X'_2 = X_2 \overline{X_2}$ and $Y'_2 = Y_2 \overline{Y_2}$. Furthermore, we will use boldface letters to denote different classical registers with the same dimensions, for example $\BX = \boldX \overline{\boldX}$ will denote an independent register of the same dimension as $\pX$. One should think of the boldface registers as a relabeling of the original registers but they will be needed since we will consider states like $\purify{\sigma}{r}{\mainreg} \otimes \purify{\rho}{r}{\alt R C}$.

Also, note that if we have two unitaries $U_{XA}$ and $V_{XB}$ that both have a classical register $X$ as control, then $U_{XA}$ and $V_{XB}$ commute (recall our convention that we omit to write tensor product with the identity operator on the remaining spaces). 

\begin{proof}[Proof of Lemma \ref{lemma:qcond}]

We will bound
\begingroup\leqnos
\begin{align}\label{eqn:uhlmann}
\    \hspace*{0.9cm} \norm{\pstatealt{\sigma}{r}{X_1Y_1 C} - \pstatealt{\sigma}{r}{X_1Y_1} \otimes \pstatealt{\rho}{r}{C}} & \le 2\sqrt{2}\,\he{\pstatealt{\sigma}{r}{X_1Y_1C }}{\pstatealt{\sigma}{r}{X_1Y_1} \otimes \pstatealt{\rho}{r}{C}} \\
\ &=  2\sqrt{2}\,\min_{\widetilde{U}} \he{\widetilde{U}\purify{\sigma}{r}{\mainreg RC} \otimes \ketalt{\rho}{\BX\BY }}{\purify{\sigma}{r}{\mainreg} \otimes \purify{\rho}{r}{\alt R C}} \notag
\end{align}\endgroup
where the inequality used Proposition \ref{prop:he-and-trace}, and the equality follows from Uhlmann's theorem (Proposition \ref{prop:uhlmann}) with $\widetilde{U}$ ranging over all unitaries acting on $\X_1\Y_1\alt R$.  Notice that apart from $\X_1\Y_1 R$, we also need the boldface registers to make the state $\pstatealt{\sigma}{r}{X_1Y_1} \otimes \pstatealt{\rho}{r}{C}$ a pure state. Also, note that  $\purify{\rho}{r}{\mainreg}$ and $\purify{\sigma}{r}{\mainreg}$ remain the same throughout all rounds, so we will drop the superscript~$r$ for these states. 

To upper bound the right-hand side in \eqref{eqn:uhlmann}, we will exhibit a unitary $\widetilde{U}$ so that the  Hellinger distance is small. Let us first note that since $u'(X_1) = q(X_1)$, using Proposition \ref{prop:control}, there exists a unitary $W_{X_1\pY}$ with $X_1$ as a control register such that $W_{X_1\pY} \ketalt{\rho}{\pX}\ketalt{\rho}{\pY}= \ketalt{\sigma}{\pX \pY}$. Similarly, since $u'(Y_1) = q(Y_1)$ there exists a similar unitary $W_{\pX Y_1}$ with $Y_1$ as a control (we will use the same letter to denote them since the subscripts will make the registers clear).

We first claim that 

\begin{claim} \label{claim:exist}
There exist unitaries $V^{(s)}_{\X_1 \BX X'_2 A}$ for odd $s$, and $V^{(s)}_{\Y_1 \BY Y'_2B}$ for even $s$ with $V^{(0)}_{\Y_1 \BY Y'_2B}= I_{\Y_1 \BY Y'_2B}$, such that 
\begin{align*}
\  \he{V^{(s)}_{\X_1 \BX X'_2 A} W_{X_1\pY}\purify{\rho}{s}{\pX \BY R C} \otimes \ketalt{\rho}{\BX \pY}}{\purify{\rho}{s}{\BX \BY R C} \otimes \ketalt{\sigma}{\pX \pY}}  &\le \sqrt{\eps_s} \text{ for odd } s{,~~and}\\
\  \he{V^{(s)}_{\Y_1 \BY Y'_2B} W_{\pX Y_1}\purify{\rho}{s}{\BX \pY RC} \otimes \ketalt{\rho}{\pX\BY}}{\purify{\rho}{s}{\BX \BY R C} \otimes \ketalt{\sigma}{\pX \pY}} &\le \sqrt{\eps_s} \text{ for even } s.
\end{align*}
\end{claim}

Below we will drop the registers when we are writing states over all the registers $\mainreg RC \alt$. We will also drop the registers from the unitaries $V^{(s)}$ since their indices (whether odd or even) will describe the corresponding registers they act on, unless we need to emphasize it.

Let us define 
\begin{align}\label{eqn:states}
     \ \purify{\theta}{s}{} = V^{(s)} V^{(s-1)} \purify{\sigma}{s}{\mainreg R C} \otimes \ketalt{\rho}{\BX\BY} \text{    and   } \purify{\lambda}{s}{} =  \purify{\rho}{s}{\alt R C} \otimes \ketalt{\sigma}{\mainreg}.
\end{align}
Then, we will prove by induction that for every round $s$, the following holds:
\begin{claim}\label{claim:induction}
$ \he{\purify{\theta}{s}{}}{\purify{\lambda}{s}{}}  \le \delta_s$ where $\delta_s = \sqrt{\eps_s} + \sqrt{\eps_{s-1}} + 2\sum_{i=1}^{s-2} \sqrt{\eps_i}.$
\end{claim}

For $s=r$, Claim \ref{claim:induction} implies that the unitary $V^{(r)}V^{(r-1)}$ is a particular unitary acting on $\X_1\Y_1\alt R$ for which the right-hand side in \eqref{eqn:uhlmann} is small, so taking $\widetilde{U}$ to be $V^{(r)}V^{(r-1)}$ in \eqref{eqn:uhlmann},
\[ \norm{\pstatealt{\sigma}{r}{X_1 Y_1 C} - \pstatealt{\sigma}{r}{X_1 Y_1} \otimes \pstatealt{\rho}{r}{C}}  \le 2 \sqrt{2}\left(\sqrt{\eps_r} + \sqrt{\eps_{r-1}} + 2\sum_{s=1}^{r-2} \sqrt{\eps_s} \right) \le 4 \sqrt{2} \left(\sum_{s=1}^r \sqrt{\eps_s}\right).\] 
This completes the proof of Lemma~\ref{lemma:qcond} assuming the claims. 
\end{proof}
We next prove Claims~\ref{claim:exist} and \ref{claim:induction} in order.
\begin{proof}[Proof of Claim \ref{claim:exist}] 
We will only prove the first inequality as the second one is analogous. From the assumption that $\Inf[\rho^{(s)}]{X_1}{Y\YR BC} \le \eps_s$ it also follows that $\Inf[\rho^{(s)}]{X_1}{\BY Y'_2 BC} \le \eps_s$ since we are just relabeling the $\pY$ registers to $\BY$ (recall $\pY = Y_1 \overline Y_1$). Using Proposition \ref{prop:he-and-mutinf},  
\[ \he{\pstatealt{\rho}{s}{X_1\BY Y'_2 BC}}{\pstatealt{\rho}{s}{\BY Y'_2 BC} \otimes \rho_{X_1}} \le \sqrt{\Inf[\rho^{(s)}]{X_1}{\BY Y'_2 BC}} \le \sqrt{\eps_s}.\]
Recalling that $R = \rest$ and using Uhlmann's Theorem (Proposition \ref{prop:uhlmann}), there exists a unitary $V^{(s)}_{\X_1 \BX X'_2 A}$ such that 
\begin{align*}
    \ \sqrt{\eps_s} &\ge \he{V^{(s)}_{\X_1 \BX X'_2 A} \purify{\rho}{s}{\pX \BY R C} \otimes \ketalt{\rho}{\BX}}{\purify{\rho}{s}{\BX \BY R C} \otimes \ketalt{\rho}{\pX}} \notag \\
    \           &=\he{V^{(s)}_{\X_1 \BX X'_2 A} \purify{\rho}{s}{\pX \BY R C} \otimes \ketalt{\rho}{\BX} \otimes \ketalt{\rho}{\pY}}{\purify{\rho}{s}{\BX \BY R C} \otimes \ketalt{\rho}{\pX} \otimes \ketalt{\rho}{\pY}} \\
    \  &= \he{V^{(s)}_{\X_1\BX X'_2 A} W_{X_1\pY}\purify{\rho}{s}{\pX \BY R C} \otimes \ketalt{\rho}{\BX\pY}}{\purify{\rho}{s}{\BX \BY R C} \otimes \ketalt{\sigma}{\pX\pY}},
\end{align*}
where in the last equality we multiplied both states by the unitary $W_{X_1\pY}$ and used that $W_{X_1\pY} \ketalt{\rho}{\pX} \otimes \ketalt{\rho}{\pY} = \ketalt{\sigma}{\pX\pY}$ as well as the fact that $W_{X_1\pY}$ and $V^{(s)}_{\X_1 \BX X'_2 A}$ commute (disjoint registers).
\end{proof} 

\begin{proof}[Proof of Claim \ref{claim:induction}]
\underline{Base case $s=1$:} Recall that $V^{(0)}$ is the identity. Let $W_{X_1\pY}$ be the unitary that satisfies $W_{X_1\pY} \ketalt{\rho}{\pX \pY}= \ketalt{\sigma}{\pX \pY}$ as before. Then, since $u'(X)=q(X)$ and $q(Y_1Y_2)=q(Y_1)q(Y_2)$ and $q(Y_2)=u'(Y_2)$, it follows from Proposition \ref{prop:qprotocol} that 
$W_{X_1\pY} \purify{\rho}{1}{\pX \pY R C} = \purify{\sigma}{1}{\pX \pY R C}$ (recall that $R = \rest$). Using this and the fact that $\purify{\rho}{1}{\pX \BY R C} = \purify{\rho}{1}{\pX R C} \otimes \ketalt{\rho}{\BY}$ and $\ketalt{\rho}{\BX \pY} = \ketalt{\rho}{\BX} \otimes \ketalt{\rho}{\pY}$, we get
\begin{align} \label{eqn:base}
    \ W_{X_1\pY} \purify{\rho}{1}{\pX \BY R C} \otimes \ketalt{\rho}{\BX \pY} &=  W_{X_1\pY} \purify{\rho}{1}{\pX \pY R C} \otimes \ketalt{\rho}{\BX \BY} = \purify{\sigma}{1}{\pX \pY R C} \otimes \ketalt{\rho}{\BX \BY}.
\end{align}

Furthermore, by the definition of $\purify{\theta}{1}{}$ and $\purify{\lambda}{1}{}$ with equations  \eqref{eqn:states} and  \eqref{eqn:base} above, it follows that

\begin{align*}
\ \he{\purify{\theta}{1}{}}{\purify{\lambda}{1}{}} &= \he{V^{(1)}_{\X_1 \BX X'_2 A} W_{X_1\pY} \purify{\rho}{1}{\pX \BY R C} \otimes \ketalt{\rho}{\BX \pY}}{\purify{\rho}{1}{\BX \BY R C} \otimes \ketalt{\sigma}{\pX \pY}}  \le \sqrt{\eps_1},
\end{align*}
where we used \eqref{eqn:base} to show that  $\purify{\theta}{1}{}$ equals the first state in the middle expression and the inequality follows from Claim \ref{claim:exist}. This proves the base case.\\

\underline{Induction:} For the induction let us assume that $s$ is even (since the case for odd $s$ is similar) and that 
$\he{\purify{\theta}{s-1}{}}{\purify{\lambda}{s-1}{}} \le \delta_{s-1}$. Using the triangle inequality we bound  
\begin{align} \label{eqn:triangle}
\he{\purify{\theta}{s}{}}{\purify{\lambda}{s}{}} \le \he{\purify{\theta}{s}{}}{\purify{\omega}{s}{}}  + \he{\purify{\omega}{s}{}}{\purify{\pi}{s}{}}  + \he{\purify{\pi}{s}{}}{\purify{\lambda}{s}{}}, 
\end{align}
where 
\begin{equation}\label{eqn:defn}
    \ \purify{\omega}{s}{} = V^{(s)} U^{(s)}_{YBC}{\left(V^{(s-2)}\right)}^\dag  \purify{\lambda}{s-1}, \text{  and } \purify{\pi}{s}{} = V^{(s)} W_{\pX Y_1}  \purify{\rho}{s}{\BX\pY R C} \otimes \ketalt{\rho}{\pX\BY},
\end{equation}
with $U^{(s)}_{YBC}$ being the protocol unitary with $Y$ as control that Bob applies in round $s$. Note that from the definition of the protocol we have that 
\begin{equation} \label{eqn:int}
\ \purify{\sigma}{s}{\pX\pY RC} = U^{(s)}_{YBC} \purify{\sigma}{s-1}{\pX\pY RC}.
\end{equation}

Let us consider the first term in \eqref{eqn:triangle}. Since Hellinger distance is unitarily invariant, we multiply both states with $V^{(s-2)} {\left(U^{(s)}_{YBC}\right)}^{\dag}{\left(V^{(s)}\right)}^\dag$ and using  \eqref{eqn:states}, \eqref{eqn:defn} and \eqref{eqn:int}, we get that
\begin{align*}
\ \he{\purify{\theta}{s}{}}{\purify{\omega}{s}{}} \stackrel{\eqref{eqn:defn}}{=} \he{V^{(s-2)} {\left(U^{(s)}_{YBC}\right)}^{\dag}{\left(V^{(s)}\right)}^\dag\purify{\theta}{s}{}}{\purify{\lambda}{s-1}{}} = \he{\purify{\theta}{s-1}{}}{\purify{\lambda}{s-1}{}} \le \delta_{s-1}, 
\end{align*}
where we used that the first state in the middle expression equals $\purify{\theta}{s-1}{}$:
\begin{align*}
\ V^{(s-2)} {\left(U^{(s)}_{YBC}\right)}^{\dag}{\left(V^{(s)}\right)}^\dag\purify{\theta}{s}{}   & \stackrel{\eqref{eqn:states}}{=} V^{(s-2)} {\left(U^{(s)}_{YBC}\right)}^{\dag}V^{(s-1)}\purify{\sigma}{s}{\pX \pY RC} \otimes \ketalt{\rho}{\BX\BY}  \\
\ & \stackrel{\eqref{eqn:int}}{=} V^{(s-2)} {V^{(s-1)}} \purify{\sigma}{s-1}{\pX\pY RC} \otimes \ketalt{\rho}{\BX\BY} \\
\ &= V^{(s-1)}{V^{(s-2)}}\purify{\sigma}{s-1}{\pX\pY RC} \otimes \ketalt{\rho}{\BX\BY} \stackrel{\eqref{eqn:states}}{=} \purify{\theta}{s-1}{},
  \end{align*}
with the second equality using \eqref{eqn:int} and the fact that $U^{(s)}_{YBC}$ and $V^{(s-1)}_{\X_1 \BX X'_2 A}$ commute, and the third equality using that $V^{(s-1)}_{\X_1 \BX X'_2 A}$ and $V^{(s-2)}_{\Y_1 \BY Y'_2 B}$ commute.

To bound the second term, notice that by the definition of the protocol $ \purify{\rho}{s-1}{\BX\BY RC} = U^{(s-1)}_{\boldX X_2 AC} \purify{\rho}{s-2}{\BX\BY RC}$ (recall $R = \rest$) and therefore using \eqref{eqn:defn} and \eqref{eqn:states}, it follows that
\begin{equation} \label{eqn:int1}
\ \purify{\omega}{s}{} = V^{(s)} U^{(s)}_{YBC}{(V^{(s-2)})}^\dag U^{(s-1)}_{\BX X_2 AC} \purify{\rho}{s-2}{\BX\BY RC} \otimes \ketalt{\sigma}{\pX\pY}.
\end{equation}
Now multiplying both states by $Z = \left(U^{(s-1)}_{\boldX X_2 AC}\right)^\dag V^{(s-2)} {\left(U^{(s)}_{YBC}\right)}^{\dag}  {\left(V^{(s)}\right)}^\dag$ and again using unitary invariance of Hellinger we get that
\begin{align*}
\ \he{\purify{\omega}{s}{}}{\purify{\pi}{s}{}} &= \he{Z\purify{\omega}{s}{}}{Z\purify{\pi}{s}{}} \\
\ & \stackrel{\eqref{eqn:int1}}{=} \he{\purify{\rho}{s-2}{\alt R C} \otimes \ketalt{\sigma}{\mainreg}}{V^{(s-2)}\purify{\rho}{s-2}{\BX\pY R C} \otimes \ketalt{\sigma}{\pX \BY}} \le \sqrt{\eps_{s-2}},
\end{align*}
where the inequality follows from Claim \ref{claim:exist} and we simplified the second state $Z\purify{\pi}{s}{}$ using commutativity of the pairs $\left\{{\left(V^{(s-2)}_{\Y_1 \BY Y'_2 B}\right)}^{\dag}, \left(U^{(s-1)}_{\boldX X_2 AC}\right)^\dag\right\}$ (disjoint registers), $\left\{{\left(U^{(s)}_{YBC}\right)}^{\dag} , W_{\pX Y_1}\right\}$ (disjoint registers except for both being controlled on the shared register $Y_1$), and $\left\{{\left(U^{(s-1)}_{\boldX X_2AC}\right)}^{\dag} , W_{\pX Y_1}\right\}$ (disjoint registers) as follows:
\begin{align*}
\ Z \purify{\pi}{s}{} &= \left(U^{(s-1)}_{\boldX X_2 AC}\right)^\dag V^{(s-2)} {\left(U^{(s)}_{YBC}\right)}^{\dag}  {\left(V^{(s)}\right)}^\dag \purify{\pi}{s}{} \\
\ & \hspace*{-0.75cm}  \stackrel{\eqref{eqn:defn}}{=} V^{(s-2)}  \left(U^{(s-1)}_{\boldX X_2 AC}\right)^\dag {\left(U^{(s)}_{YBC}\right)}^{\dag} W_{\pX Y_1} \purify{\rho}{s}{\BX\pY R C} \otimes \ketalt{\rho}{\pX\BY} \\
\ & \hspace*{-0.75cm}  = V^{(s-2)}  \left(U^{(s-1)}_{\boldX X_2 AC}\right)^\dag W_{\pX Y_1} {\left(U^{(s)}_{YBC}\right)}^{\dag}  \purify{\rho}{s}{\BX\pY R C} \otimes \ketalt{\rho}{\pX\BY} \\
\ & \hspace*{-0.75cm} =V^{(s-2)}  W_{\pX Y_1} \left(U^{(s-1)}_{\boldX X_2 AC}\right)^\dag \purify{\rho}{s-1}{\BX\pY R C} \otimes  \ketalt{\rho}{\pX\BY}=V^{(s-2)} W_{\pX Y_1} \purify{\rho}{s-2}{\BX\pY R C} \otimes \ketalt{\rho}{\pX\BY}.
\end{align*}

To upper bound the third term of~\eqref{eqn:triangle}, by the definition of states $\purify{\pi}{s}{}$ and $\purify{\lambda}{s}{}$ (equations \eqref{eqn:defn} and \eqref{eqn:states}) and Claim~\ref{claim:exist}, we get 
\[ \he{\purify{\pi}{s}{}}{\purify{\lambda}{s}{}} = \he{V^{(s)} W_{\pX Y_1} \purify{\rho}{s}{\BX\pY R C} \otimes \ketalt{\rho}{\pX \BY} }{\purify{\rho}{s}{\alt R C} \otimes \ketalt{\sigma}{\mainreg}} \le \sqrt{\eps_{s}}. \]

Plugging the bounds for each of the terms back in \eqref{eqn:triangle}, we get that 
 \[ \he{\purify{\theta}{s}{}}{\purify{\omega}{s}{}} \le \delta_{s-1} + \sqrt{\eps_{s-2}} + \sqrt{\eps_s} = \sqrt{\eps_s} + \sqrt{\eps_{s-1}} + 2\sum_{i=1}^{s-2} \sqrt{\eps_i} = \delta_s. \qedhere\]
\end{proof}

\end{document}